\documentclass[letterpaper]{article} %
\usepackage{times}  %
\usepackage{helvet}  %
\usepackage{courier}  %
\usepackage[hyphens]{url}  %
\usepackage{graphicx} %
\usepackage{natbib}  %
\usepackage{caption} %
\usepackage{comment}
\usepackage{enumerate}

\usepackage{amsthm,amsmath,amssymb}
\usepackage{cleveref}
\theoremstyle{definition}

\usepackage{tikz}
\usetikzlibrary{automata, positioning, calc, shapes, arrows, fit}

\usepackage{pgfplots}
\usepackage{booktabs}

\DeclareMathOperator{\rank}{rank}

\theoremstyle{plain}

\newtheorem{example}{Example}

\newtheorem{definition}{Definition}

\usepackage{thmtools}
\usepackage{thm-restate}

\usepackage{subfig}

\usepackage[ruled,vlined,noend]{algorithm2e}
\SetKw{Break}{break}

\newcommand{\reals}{\mathbb{R}}

\usepackage[backgroundcolor=orange!15!white,bordercolor=orange!80!black,textsize=tiny]{todonotes} 
\tikzset{/tikz/notestyleraw/.append style={rounded corners=0pt,inner sep=0.6ex}}
\presetkeys{todonotes}{inline}{}

\title{Proportional Representation in Practice:\\
Quantifying Proportionality in Ordinal Elections
}

\author{
Tuva Bardal\footnote{University of Warwick, Coventry, UK; email: tuva.bardal@warwick.ac.uk}
\and Markus Brill\footnote{University of Warwick, UK; email: markus.brill@warwick.ac.uk} 
\and David McCune \footnote{William Jewell College, Liberty, MO, United States; email: mccuned@william.jewell.edu}
\and Jannik Peters\footnote{National University of Singapore, Singapore; email: peters@nus.edu.sg} }

\begin{document}

\maketitle

\begin{abstract}
Proportional representation plays a crucial role in electoral systems. In ordinal elections, where voters rank candidates based on their preferences, the \textit{Single Transferable Vote (STV)} is the most widely used proportional voting method. STV is considered proportional because it satisfies an axiom requiring that large enough ``solid coalitions'' of voters are adequately represented. Using real-world data from local Scottish elections, we observe that solid coalitions of the required size rarely occur in practice. This observation challenges the importance of proportionality axioms and raises the question of how the proportionality of voting methods can be assessed beyond their axiomatic performance. We address these concerns by developing quantitative measures of proportionality. We apply these measures to evaluate the proportionality of voting rules on real-world election data. Besides STV, we consider \textit{SNTV}, the \textit{Expanding Approvals Rule}, and \textit{Sequential Ranked-Choice Voting}. 
We also study the effects of ballot truncation by artificially completing truncated ballots and comparing the proportionality of outcomes under complete and truncated ballots.
\end{abstract}

\section{Introduction}

Proportional representation is a core principle of modern electoral systems, its goal being to ensure that elected officials proportionally represent the composition of the electorate. Most countries achieve proportional representation through party-list systems, where voters cast their ballots for political parties rather than individual candidates, and legislative seats are allocated in proportion to the number of votes each party receives. However, several countries---mostly from the former British commonwealth---use a system in which voters can vote for individual candidates using ranked preferences. In such settings, the electorate is usually divided into a number of districts, each of which elects a few candidates to represent the district in the given governing body. To achieve proportional representation in district elections which do not use a party-list system, most localities use the \emph{Single Transferable Vote (STV)}, a voting rule dating back to the 19th century \citep{Tide95a}. STV is generally considered to guarantee proportional representation since it satisfies a property known as \emph{proportionality for solid coalitions (PSC)} formulated by \citet{Dumm84a}. This property, in essence, guarantees that any sufficiently large group of voters that rank the same candidates (not necessarily in the same order) above all other candidates is given an amount of representation commensurate with its size. Such groups are referred to as \textit{solid coalitions}.

Despite the prominence of PSC in the literature, some have challenged the idea that PSC is a good way to capture  the notion of proportional representation rigorously. For instance, the requirement that all voters in a solid coalition must all share the exact same candidate set in their ballot prefix makes PSC highly non-robust \citep{Tide06a}. %
This issue has been considered from a purely theoretical point of view \citep{AzLe20a, BrPe23a}, as well as through experiments on synthetic data \citep{BrPe23a}. There exist several other properties that seek to guarantee proportional representation in ways similar to that of PSC  \citep{AEF+17a, BrPe23a}, and that, from the aforementioned perspective, are more robust \citep{BrPe23a}. However, such properties are purely qualitative, and like PSC they usually rely rigidly on a ``threshold of representation,'' making them blind to groups of voters that come close to being ``sufficiently large.''

While proportionality axioms have received much attention in the theoretical literature, empirical investigation of the force of axioms formulated to guarantee proportional representation has been limited due to a scarcity of real-world ballot data. However, \citet{McGr23a} recently compiled a real-world dataset consisting of 1100 Scottish local council elections from the period 2007--2022. In these elections, STV is used to elect members for local councils in Scotland. We observe that in this data, solid coalitions of sufficient size rarely occur. Consequently, for most elections in the dataset, PSC places no or few restrictions on the winning committee, and most outcomes are therefore proportional according to the property. Since PSC often undergirds the claim that STV is proportional, the toothlessness of PSC in practice raises the question of how proportional the method really is.\footnote{As demonstrated by \citet{BrPe23a}, STV fails proportionality axioms that are stronger than PSC.} More generally, we consider the issue of how to assess the proportionality of a committee when axioms have little to no effect on the outcome. As a first step towards answering this question, we define quantitative versions of several proportionality axioms suggested in the literature, based on the idea that we should be able to relax size-requirements imposed on groups of voters whenever there are few or no groups of sufficient size. This approach not only allows us to measure the extent to which an axiom is satisfied, but also makes it possible to strengthen axioms in situations where they have little to no effect. We furthermore use the data from the Scottish local council elections to  assess experimentally the proportionality of different voting rules according to our measures.

In practice, voters tend not to form cohesive groups large enough for standard proportionality axioms to place any significant requirement on the winning committee. The issues of ballot truncation and small cohesive groups have been discussed previously in \citet{BrPe23a, MaPl16a} and \citep{HKRW21a}, but such work had access to little to no real-world ballot data. Our contribution is that we investigate how to adapt proportionality axioms to such a real-world setting, and we provide empirical results from a large dataset of real-world STV elections.

\subsection*{Related Work}
\paragraph{Proportionality.} The formal study of proportional representation in multiwinner voting dates to the work of \citeauthor{Dumm84a}, whose proportionality for solid coalitions (PSC) criterion \citep{Dumm84a} was designed to formally capture the proportionality of STV. In recent years, inspired by \citeauthor{Dumm84a}, proportional representation has been studied as a fairness criterion in several different decision-making scenarios. Most prominently, starting with \citet{ABC+16a} and their notion of \emph{extended justified representation (EJR)} a wide body of work covers proportionality in approval-based multiwinner voting \citep{LaSk22a}. This has since been used as a basis to study proportionality in settings such as clustering \citep{CFL+19a, KePe23a, KKK24a, CMS24a}, sortition \citep{EbMi23a}, participatory budgeting \citep{PPS21a, BFL+23a}, or AI-augmented civic participation platforms \citep{FGP+23a}. 

\paragraph{Proportionality with Ordinal Preferences.} As mentioned previously, PSC was initially conceptualized by \citet{Dumm84a}.  Following this \citet{AzLe20a} introduced a generalization of PSC to weak-ordinal preferences (i.e., ordinal preferences with ties allowed) and introduced the \emph{expanding approvals rule} as an alternative to STV satisfying this generalization. In a recent work, \citet{DePe24a} studied generalizations of STV (and its single-winner variant) to weak-ordinal preferences and showed that the natural generalization of STV also satisfies the PSC generalization of \citeauthor{AzLe20a}. In follow-ups, \citet{AzLe21a} additionally generalized PSC to the setting of participatory budgeting with weak-ordinal preferences and introduced a second generalization of PSC they termed IPSC and derived a characterization of PSC through what they call a Dummett tree \citep{AzLe22a}.

Deviating from the use of solid coalitions, \citet{BrPe23a} introduced proportionality axioms for ordinal preferences, based on the proportionality axioms for approval preferences and argued that they are more robust than the original PSC based axioms. These axioms were further generalized to the setting of proportional clustering by \citet{KePe23a} and \citet{ALM23a}. In a different context, \citet{AEF+17a} and \citet{JMW20a} independently (among other things) study the property of local stability, a possible extension of Condorcet consistency to multiwinner voting. Both show that committees satisfying local stability might not exist in a given election. However, \citet{JMW20a} design an approximation algorithm, showing that a $(16+\varepsilon)$-approximation to local stability always exists. This approximation factor was recently improved by \citet{CLP+24a} to 9.8217. Locally stable committees additionally found application in the design of randomized voting rules with low distortion \citep{EKPS24a}. Improving the approximation constant is an open question and additionally has implications for the existence of so-called Condorcet winning sets \citep{ELS15a, CLP+24a}.

Finally, \citet{Jans18a} also studies this measure for evaluating the proportionality of voting rules, deriving the threshold values for several voting rules in elections with approval ballots, ordinal ballots, or party lists. Additionally, the concept of the proportionality degree \citep{Skow21a} is highly related, proving a very similar proportionality guarantee in the approval world. 

\paragraph{Experiments for Multiwinner Voting.} Next, we  highlight the related work on the experimental evaluation of (proportional) multiwinner voting rules. Closely related to us is the recent work of \citet{BBC+24a} who experimentally evaluate the use of proportional multiwinner voting rules in the Polkadot blockchain, based on real-world data from this blockchain. Further, in the participatory budgeting literature, several papers have analyzed and compared rules on real-world datasets \citep{BFJK23a, FFP+23a, BFJ+24b}. Besides this, most experimental papers have so far focused on synthetic data or data not necessarily collected for the purpose of multiwinner voting; see the survey by~\citet{BFJ+24x}. For instance, \citet{EFL+17a} evaluate ordinal multiwinner voting rules on randomly selected $2$-dimensional Euclidean instances.

\paragraph{Ballot truncation} We note that the issue of ballot truncation in real-world elections has received little attention in the multiwinner ranked-choice setting, mainly due to the lack of available data. The only study we are aware of which analyzes the effects of ballot truncation empirically is \citet{HKRW21a}, who only examine six city council elections from Cambridge, Massachusetts. Ballot truncation has received more attention in the single-winner instant runoff setting. For example, \citep{KGF20, TUK23a} study the effect on the election winner as we vary the amount of truncation, and \citep{BK15, GSM23} analyze how ballot exhaustion (when partial ballots run out of preferences before the final round) can cause the election winner not to secure a majority of total votes cast.

\section{Preliminaries}
\label{sec:prelims}

For $t \in \mathbb{N}$, let $[t] = \{1, \dots, t\}$. 
Throughout the paper, we assume that we are given a set $N = [n]$ of \emph{voters} and a set $C = \{c_1, \dots, c_m\}$ of \emph{candidates}. %
The preferences of voters are \emph{top-truncated}. That is, each voter $i \in N$ chooses a set $A_i \subseteq C$ of candidates  and a \emph{strict} ranking \mbox{$\succ_i \colon A_i \times A_i$} over these candidates, where $|A_i|$ may be less than $m$. For notational purposes, we assume that $c \succ_i c'$ for any $c \in A_i$ and $c' \notin A_i$, while voters are indifferent among candidates they do not rank. Finally, let $k$ denote the number of candidates that need to be selected. 
We refer to subset $W\subseteq C$ of candidates of size $|W|=k$ as a \textit{committee}.
A \emph{(multiwinner voting) instance $I$} is a collection of voters, candidates, voter preferences, and the committee size. %

\subsection{Proportionality Axioms} \label{sec:axioms}

For $\ell \in [k]$, we say that a group  $N' \subseteq N$ of voters is \emph{$\ell$-large} if $\lvert N'\rvert \ge \ell \frac{n}{k}$. Generally speaking, in proportional multiwinner voting if $N'$ is $\ell$-large and $N'$ is ``cohesive'' in some sense, then $\ell$ of the candidates supported by $N'$ should receive seats on the winning committee.

The most prominent proportionality axiom for ranked preferences is \emph{proportionality for solid coalitions (PSC)} introduced by \citet{Dumm84a}. Given a subset $N' \subseteq N$ of voters and $C' \subseteq C$ of candidates, $N'$ forms a \emph{solid coalition} over $C'$ if for any pair of candidates $c_j \in C'$ and $c_r \in C \setminus C'$, it holds that $c_j \succ_i c_r$ for all $i \in N'$. In other words, $C'$ forms a prefix of the ranking $\succ_i$ of every voter in $N'$. Using the notion of solid coalitions, we can now define PSC.\footnote{A related axiom is \emph{generalized PSC} \citep{AzLe20a}, which generalizes PSC to the case of weak orders. For top-truncated preferences, PSC and generalized PSC are equivalent.}

\begin{definition}[PSC]
    A committee $W$ satisfies \emph{proportionality for solid coalitions (PSC)}
    if for any $\ell \in [k]$ and any $\ell$-large group $N' \subseteq N$ of voters forming a solid coalition over \mbox{$C' \subseteq C$},
    it holds that 
    $\lvert C'\cap W\rvert \ge  \min(\lvert C'\rvert, \ell)$.
\end{definition}

As a potential alternative to PSC and possible generalization of the Condorcet principle to proportional representation, \citet{AEF+17a} introduced the concept of \emph{local stability}. Intuitively, local stability postulates that no group of voters of size at least $\frac{n}{k}$ should find an unselected candidate they all prefer to everyone in the committee.\footnote{The same concept was independently studied by \citet{JMW20a} in the more general context of core stability.} 
Notably, committees satisfying local stability need not exist.
\begin{definition}[LS] %
    A committee $W$ satisfies \emph{local stability (LS)} if there is no $1$-large group of voters $N' \subseteq N$  and candidate $c \notin W$ with 
    $c \succ_i c'$ for all ${i \in N'}$ and $c' \in W$.
\end{definition}

Observing that PSC places minimal restriction on the winning committee when confronted with instances in which few solid coalitions exist, \citet{BrPe23a} introduced ``rank-based'' axioms (based on axioms in approval-based multiwinner voting) that strengthen to conditions imposed by PSC.\footnote{Since we are only dealing with ordinal preferences in this paper, we omit the ``rank''-prefix used by \citet{BrPe23a}.} The strongest such axiom that can always be satisfied is PJR+.
Intuitively, if a group of voters deserving $\ell$ candidates all rank a candidate at most at rank $r$ and this candidate is not selected, then $\ell$ candidates ranked by someone in this group at rank $r$ or better should be included in the committee. 

\begin{definition}[PJR+ \citep{BrPe23a}]
    A committee $W$ satisfies \emph{PJR+} if there is no $\ell \in [k]$, $\ell$-large group $N'\subseteq N$ of voters, unselected candidate $c \notin W$, and rank $r \in [m]$ such that
    \begin{enumerate}[(i)]
        \item $\rank(i,c) \le r $ for all $i \in N'$
        \item $\lvert \{c' \colon \rank(i,c') \le r \text{ for some } i \in N'\} \cap W \rvert < \ell$.
    \end{enumerate}
\end{definition}
\citet{BrPe23a} were able to show that PJR+ can be verified in polynomial time via a reduction to submodular function minimization. This, however, is more of a theoretical result, as existing polynomial-time submodular function minimization are both difficult to implement and slow, taking upwards of $\Omega(n^{4})$ time to run \citep{Cor07a}. Therefore, we refrain from testing PJR+ in our instances and turn to some simpler to compute alternatives, also based on approval-based multiwinner voting. We start with the axiom EJR+ as defined by \citet{BrPe23a}. Note that EJR+ is not always be satisfiable.
\begin{definition}[EJR+] %
    A committee $W$ satisfies \emph{EJR+} if there is no $\ell \in [k]$, $\ell$-large group $N'\subseteq N$ of voters, unselected candidate $c \notin W$, and rank $r \in [m]$ such that
    \begin{itemize}
        \item[i)] $\rank(i,c) \le r $ for all $i \in N'$
        \item[ii)] $\lvert \{c' \colon \rank(i,c') \le r\} \cap W \rvert < \ell$ for all $i \in N'$.
    \end{itemize}
\end{definition}

As an alternative to EJR+ we take \emph{priceability} as defined by \citet{BrPe23a} and \citet{PeSk20a}. Priceability is always satisfiable, for instance by the expanding approvals rule, and in essence tries to compute a fractional matching between voters and candidates. 
\begin{definition}[Priceability]
    A committee $W$ is priceable if for each voter $i \in N$ there is a price function $p_i \colon C \to [0,1]$ and a price $p \in \mathbb{R}^+$ such that
    \begin{itemize}
        \item $\sum_{c \in C} p_i(c) \le 1 $ for all $i \in N$
        \item $\sum_{i \in N} p_i(c) \le p$ for each $c \in W$
        \item $\sum_{i \in N} p_i(c) = 0$ for each $c \notin W$
        \item $\sum_{i \in N\colon \rank(i,c) \le r} (1 - \sum_{c' \in W\colon \rank(i,c') > r} p_i(c')) \le p$ \\ \quad \quad for all $r \in [m]$.
    \end{itemize}
\end{definition}
We minimally deviate here from the original definition of \citet{PeSk20a} by requiring that at most $p$ can be paid for a candidate (instead of exactly $p$). This allows us to associate a price $p$ with any committee. In particular, a price $p \ge \frac{n}{k}$ is now also possible, which will allow us to use priceability to quantify the proportionality of ``disproportional'' committees as well. Using the same proof as \citet{BrPe23a} it follows that if a committee is priceable with a price $p < \frac{n}{k}$ the committee also satisfies PJR+.

\subsection{Voting Rules} \label{sec:rules}

A voting rule maps each instance to one or more \textit{winning committees}. 
We briefly introduce the voting rules we study. For more details, we refer to the full version of this paper.

The \emph{Single Transferable Vote (STV)} is a family of rules, with different versions of STV used in different jurisdictions \citep{Tide95a}. Following \citet{McGr23a}, we describe  the version that is used in Scottish local elections. 
STV proceeds in rounds and starts by assigning each voter $i \in N$ a weight $w_i = 1$. In each round, STV checks whether the candidate with the most weighted first-place votes has at least $q$ (weighted) first-place votes, where $q = \lfloor \frac{n}{k+1}\rfloor + 1$ is the \textit{quota}. If that is the case, this candidate is elected and the voters ranking this candidate first are reweighted proportionally such that the total weight of the election decreases by exactly $q$. If there is no such candidate, the candidate with the least weighted first-place votes is removed. Besides this version of STV, which we refer to as as \emph{Scottish STV}, we also consider \emph{Meek-STV} \citep{HWW87x}.

The \emph{Expanding Approvals Rule (EAR)} is another family of proportional multiwinner voting rules \citep{AzLe20a}. Just like STV, it uses a quota $q$ and starts by assigning each voter a weight $w_i = 1$. It then iterates over all possible ranks from first to last. For each such rank and so-far unselected candidate $c$ it checks whether the total weight of the voters giving a rank of at most $r$ to $c$ is at least~$q$. If there is such a candidate, it takes any such candidate into the committee and decreases the collective weight of the corresponding voters by~$q$. 
In our implementation of the rule we select the candidate with the largest total weight and decrease the weights as in Scottish STV. 
If there is no such candidate, $r$ is increased.\footnote{The treatment of unranked candidates in EAR allows for different interpretations, as noted in Remark 3 by \citet{AzLe20a}. When a voter ranks a strict subset $A_i$ of candidates, the set of unranked candidates $C \setminus A_i$ (i.e., the last equivalence class) can be included either \textit{(i)} as soon as the ranked candidates are exhausted, or \textit{(ii)} only in the final step of the method. We tested both variants in our experimental analysis. The performance differences between the two versions were minor, with variant \textit{(ii)} performing slightly better. Therefore, we focus on variant \textit{(ii)} here.}

The \emph{Single Non-Transferable Vote (SNTV)} selects the $k$ candidates with the highest first-place vote count. SNTV is sometimes referred to as \textit{$k$-plurality}.

\emph{Sequential Ranked-Choice Voting (seq-RCV)}, which is currently used in the US state of Utah \citep{MMLS24a}, executes the single-winner RCV procedure $k$ times. Single-winner RCV (a.k.a. \textit{instant runoff voting}) iteratively deletes the candidate with the fewest first-place votes until only a single candidate is left. 

\smallskip 

We include seq-RCV in the rules we examine because it is not proportional but \textit{majoritarian} (i.e., a group consisting of barely more than half of the electorate can force the entire committee to consist of their candidates) and therefore can provide context for our results around methods like STV. We note that across all 1070 multiwinner Scottish elections in our dataset there are only nine in which seq-RCV chooses a winner set incompatible with PSC. Thus, if PSC is the standard by which a rule is judged to be proportional, seq-RCV is virtually proportional in practice. However, seq-RCV often produces outcomes which are wildly non-proportional under any intuitive notion of ``proportional'' \citep{MMLS24a}, and this provides additional motivation for why we should explore alternatives to standard PSC in practice.

While SNTV does not satisfy any of the  axioms discussed in \Cref{sec:axioms}, it is considered a ``semi-proportional'' method \citep{Amy00x}. Such methods aim to give some representation to minorities, albeit not proportional to their support. %

The remaining three voting rules are proportional: both versions of STV satisfy PSC, and EAR satisfies saitsfies PJR+. Furthermore, it was shown by \citet{BrPe23a} that committees returned by EAR are always priceable. %

\begin{table}[tb]
    \centering
\begin{tabular}{@{}lrrr@{}}
\toprule
 & $<25\%$\phantom{xx} & $<50\%$\phantom{xx} & $<100\%$\phantom{xx} \\ 
\cmidrule(r){2-2}\cmidrule(lr){3-3}\cmidrule(lr){4-4} %
PSC & 42 (3.9\%) & 275 (25.7\%) & 776 (72.5\%)\\ %
EJR+ & 64 (6.0\%) & 357 (33.4\%) & 1005 (93.9\%) \\ %
LS & 173 (16.2\%) & 650 (60.7\%) & 1061 (99.1\%)  \\ %
Priceability & 61 (5.7\%) & 354 (33.1\%) & 1003 (93.7\%)  \\ 
\bottomrule
\end{tabular}%
    \caption{
    For each axiom, the row values correspond to the number of elections in the dataset where the satisfaction rate is strictly less than the percentage value at the top of the column. For instance, in only 42 out of 1070 elections it is the case that less than 25\% of all committees satisfy PSC.}
    \label{tab:committees-axioms}
\end{table}

\section{Scottish Local Council Elections}
\label{sec: data}
For the purposes of local governance, Scotland is partitioned into $32$ council areas, each of which is governed by a separate council. Each council area is divided into wards, and each ward elects a number of councilors to represent the ward on the council. The number of candidates running and the number of seats available in a typical election are not large; most elections satisfy $m \in \{6,7,8,9\}$ and $k \in \{3,4\}$. Since 2007, all wards have used Scottish STV to choose their representatives. Elections are held every five years. 

\citet{McGr23a} collected ballot data from 1100 Scottish local council elections between the years 2007 and 2022.\footnote{The data is available at  \url{https://github.com/mggg/scot-elex}.} 
Out of the 1100 elections, 1070  satisfy \mbox{$k > 1$}; we only consider these 1070 instances. Notably, voters are not required to provide full rankings and ballots are often heavily truncated: across the 5,485,379 total ballots cast from all elections, approximately 14\% rank only a single candidate, and a majority, 58\%, rank fewer than~$k$. In contrast, only 13\% of ballots are complete (where by ``complete'' we mean a ballot that contains  $m-1$ or $m$ candidates). We refer to \citet{McGr23a} for more details about the dataset.

We evaluated the force of proportionality axioms on the ballot data from the elections by calculating the number of outcomes excluded by the axioms on each instance. Out of 1070 elections, there are 294 ($27.5\%$) for which every committee of size $k$ satisfies PSC, 592 ($55.3\%$) where there is only one solid coalition which earns a single seat under PSC, and 184 ($17.2\%$)  where multiple solid coalitions are deserving of seats by PSC. Thus, in real-world elections PSC does not place significant restrictions on the winning committee. The other axioms we consider are generally more discerning than PSC, however none of the axioms identify a unique outcome on any of the elections we consider. We give an example to illustrate both how PSC may fall short in excluding outcomes and how the number of compatible committees may differ between the axioms. 
An overview of the number of outcomes that satisfy each of the axioms considered over all elections in the dataset can be found in \Cref{tab:committees-axioms}.

\begin{example}
Consider the 2012 council election of Midlothian, ward 2, with $n = 5132$ voters, $m=7$ candidates, and $k=3$ seats. The  candidates, their party affiliations, and their first-place vote counts are listed %
in \Cref{tab:election_results_midl}.
\begin{table}[h!]
\centering
\begin{tabular}{llr}
\toprule
\textbf{Candidate} & \textbf{Party} & \textbf{First-Place Votes} \\
\midrule
D. Milligan (\textbf{DM})& Labour & 1,574 \\
L. Milliken (\textbf{LM})& Labour & 525 \\
J. Aitchison (\textbf{JA}) & Independent & 382 \\
B. Constable (\textbf{BC}) & SNP & 1,257 \\
T. Munro (\textbf{TM}) & SNP & 358 \\
I. Baxter (\textbf{IB}) & Greens & 671 \\
E. Cummings (\textbf{EC}) & Conservative & 365 \\
\bottomrule
\end{tabular}
\caption{Candidates and vote counts for \Cref{ex:mloth20122}.}
\label{tab:election_results_midl}
\end{table}

There are $\binom{m}{k} = \binom{7}{3} = 35$ possible outcomes in this election. The largest (non-trivial) solid coalition is over the two Labour candidates \textbf{DM} and \textbf{LM} and has size $1624$. This coalition consists of the $1218$ voters who cast a ballot of the form $\textbf{DM}\succ \textbf{LM} \succ \dots$ and of the $406$ voters who cast a ballot of the form $\textbf{LM} \succ \textbf{DM} \succ \dots$. Interestingly, the size of this coalition is much smaller than the total number of voters who ranked a Labour candidate first. The reasons are that some Labour voters rank only one candidate on their ballots and  many voters  cast split-ticket ballots. The next largest solid coalition has size $1277$ and is over the two SNP candidates \textbf{BC} and \textbf{TM}, again barely more than the first-place votes of \textbf{BC}. The largest solid coalition over more than two candidates consists of $554$ voters who support the two SNP candidates as well as \textbf{IB}. Solid coalitions over four or more candidates are extremely small. The threshold for a coalition to be $1$-large is $\lceil \frac{n}{k} \rceil = 1711$, and therefore any of the 35 possible winning committees satisfies PSC. In comparison, $24$ out of $35$ committees satisfy rank-EJR+ and priceability, while $12$ out of $35$ outcomes are locally stable. 
\label{ex:mloth20122}
\end{example}

\section{Quantifying Proportionality}
\label{sec:qProp}

In this section, we turn proportionality axioms into quantitative proportionality measures. 
The main idea behind the construction of measures consists in 
(1) defining a parameterized version of the proportionality axiom by introducing a multiplicative factor on the size constraint, and 
(2) identifying the smallest parameter for which the parameterized axiom is satisfied by the given committee. 
We start with PSC and consider other axioms in \Cref{sec:qOther}.

\subsection{Quantifying PSC}
\label{sec:qPSC}

Recall that a group $N'\subseteq N$ of voters is called $\ell$-large if $|N'|\ge \ell \frac{n}k$. 
A parameterized version of this notion can be obtained by introducing a multiplicative factor.
\begin{definition}
    Consider an instance with $n$ voters and committee size $k$.
    For $\alpha \in \reals^+$ and $\ell \in [k]$, a group $N'\subseteq N$ of voters is \emph{$\ell_\alpha$-large} if $|N'|\ge \alpha \cdot \ell \frac{n}k$.
\end{definition}
That is, the value of the parameter $\alpha$ changes the size requirement a group needs to fulfil in order to be deemed worthy of $\ell$ representatives. If $\alpha<1$, then the size constraint is relaxed, as groups of size smaller than $\ell \frac{n}k$ deserve $\ell$ representatives. On the other hand, if $\alpha>1$, then a group of voters must have larger size to deserve representation. PSC requires that $\ell$-large groups need to be represented appropriately; consequently, replacing ``$\ell$-large'' with ``$\ell_\alpha$-large'' in the definition of PSC makes the axiom more demanding for $\alpha<1$ and less demanding for $\alpha>1$.

\begin{definition}[$\alpha$-PSC] \label{def:aPSC}
    Let $\alpha \in \reals^+$. 
    A committee $W$ satisfies \emph{$\alpha$-PSC}
    if for any $\ell \in [k]$ and any $\ell_\alpha$-large subset $N' \subseteq N$ of voters forming a solid coalition over $C' \subseteq C$,
    it holds that 
$\lvert C' \cap W\rvert \ge \min\left(\lvert C'\rvert,\ell\right).$
\end{definition}

\newcommand{\apsc}{\alpha_\text{PSC}}

Clearly, $1$-PSC is equivalent to (original) PSC and lowering the value of $\alpha$ makes the axiom more demanding: If $\alpha_1 \le \alpha_2$, then $\alpha_1$-PSC implies $\alpha_2$-PSC. 
For a given committee $W$, we are therefore interested in the smallest value of $\alpha$ such that $W$ satisfies $\alpha$-PSC. Formally, let\footnote{We use infimum rather than minimum since the set of~values for which $\alpha$-PSC holds is an open interval of the form $(\apsc,+\infty)$.}
\[ \apsc(W) = \inf \{\alpha\colon  W \text{ satisfies $\alpha$-PSC} \}.\]
We refer to $\apsc(W)$ as the \textit{PSC value of $W$}.
Observe that $W$ satisfies PSC if and only if $\apsc(W)<1$.

Furthermore, we call the minimum achievable $\alpha$-value for an instance $I$ the \textit{PSC value of $I$} and denote it with
\[\apsc^*(I) = \min_{W\subseteq C\colon |W|=k} \apsc(W).\]

Choosing a winning committee which achieves $\apsc^*(I)$ might be normatively desirable in a proportional context because such a committee fulfills the spirit of PSC in the absence of large solid coalitions.

\begin{example}
    Consider again the instance from \Cref{ex:mloth20122}. %
    Here, $1624$ voters form a solid coalition over $\{\textbf{DM},\textbf{LM}\}$, followed by solid coalitions of size $1574$, $1277$, $1257$, $671$, over $\{\textbf{DM}\}$, $\{\textbf{BC},\textbf{TM}\}$, $\{\textbf{BC}\}$, and $\{\textbf{IB}\}$, respectively. The $\alpha$-values at which the solid coalitions become $\ell_\alpha$-large for $\ell \in \{1,2\}$ are given in the table below. %
\[
\begin{tabular}{ c  c c c c c }
\toprule
    $\ell$ & $\{\textbf{DM},\textbf{LM}\}$ & $\{\textbf{DM}\}$ & $\{\textbf{BC},\textbf{TM}\}$ & $\{\textbf{BC}\}$ & $\{\textbf{IB}\}$ \\
\midrule
    $1$ & 0.949 & 0.920 & 0.746 & 0.734 & 0.392 \\
    $2$ & 0.474 & -- & 0.373 & -- & -- \\
\bottomrule
\end{tabular}
\label{tab:coalitions}
\]

Consider the committees $W = \{\textbf{DM}, \textbf{LM}, \textbf{BC}\}$ and $W' = \{\textbf{DM}, \textbf{BC}, \textbf{IB}\}$. Committee $W$ is chosen by Meek-STV and EAR, while $W'$ is chosen by Scottish STV. Both committees satisfy PSC, however we can distinguish the committees based on their PSC values: $\apsc(W) = 0.392$ and $\apsc(W') = 0.474$. In fact, $W$ achieves the PSC value for the  instance (i.e., $\apsc^*(I)= 0.392$), as any smaller  value would additionally force candidate $\textbf{IB}$ to be included. From the point of view of solid coalitions, $W$ is perhaps the better choice of winning committee because the solid coalition $\{\textbf{DM},\textbf{LM}\}$ is more than double the size of any solid coalition containing $\textbf{IB}$.

\end{example}

\label{sec:compPSC}

\subsection{Computing the PSC Value of a Committee}
\label{sec:comp-apscW}

As already noted, decreasing the value of $\alpha$ leads to more representation demands by solid coalitions. 
We can identify exactly the values of $\alpha$ that makes a group $\ell_\alpha$-deserving: 
\[N' \text{ is $\ell_\alpha$-large} \quad \Leftrightarrow \quad \alpha \le \frac{|N'|}{n} \cdot \frac{k}{\ell}.\]
Let $\alpha^\ell_{(N',C')} = \frac{|N'|}{n} \cdot \frac{k}{\ell}$ denote the value of $\alpha$ for which the solid coalition $(N',C')$ becomes $\ell_\alpha$-large. 
Then, the values
\[\alpha^1_{(N',C')}, \, \alpha^2_{(N',C')}, \, \dots \,, \, \alpha^{|C'|}_{(N',C')}\]
are exactly the thresholds of $\alpha$-values for which the group starts to become deserving of $1, 2, \dots, |C'|$ many representatives under $\alpha$-PSC. (Values $\alpha^\ell_{(N',C')}$ with $\ell > |C'|$ are irrelevant because the group's deservingness is upper bounded by $|C'|$ according to the definition of $\alpha$-PSC.)

\newcommand{\soli}{\mathcal{S}}

In our algorithm for computing the PSC value of a committee, we  compute these values for all solid coalitions.
For each subset $C'\subseteq C$, there is a unique \textit{maximal} group $N_{C'}$ of voters that solidly supports $C'$. The group $N_{C'}$ consists of all voters ranking all candidates in $C'$ over all other candidates. Clearly, it is sufficient to consider only maximal solid coalitions. 
Let $\soli$ denote the set of all maximal solid coalitions.
It is not hard to see that $|\soli|$ is polynomial in the size of the profile and that we can efficiently enumerate all maximal solid coalitions by iterating over the prefixes of the voters.

Given the set $\soli$ of all maximal solid coalitions, we can now collect all threshold values for $\alpha$. 
Define $T$ as the set that contains the relevant values for each solid coalition, i.e.,
\[T \, = \bigcup_{(N',C') \in \soli} \{\alpha^1_{(N',C')}, \alpha^2_{(N',C')}, \dots, \alpha^{|C'|}_{(N',C')}\}\text.\]
Here, each threshold value $\alpha^\ell_{(N',C')}$ is associated with a {``PSC constraint''} of the form 
$|W \cap C'| \ge \ell$.

\begin{restatable}{thm}{psc-W}
Given an instance and a committee $W$, the PSC value of $W$ can be computed in polynomial time. 
\label{thm:psc-W}
\end{restatable}

\begin{proof}
First calculate the set $\soli$ of all maximal solid coalitions and the set $T$ of relevant thresholds. Consider the threshold values in $T$ in non-increasing order. 
When considering $\alpha^\ell_{(N',C')}$, check whether $|W \cap C'| \ge \ell$. If yes, go to the next threshold value. If not, we know that $\apsc(W) = \alpha^\ell_{(N',C')}$, because $\alpha^\ell_{(N',C')}$ is the largest value of $\alpha$ for which the corresponding PSC constraint is not satisfied by $W$.  
\end{proof}

\subsection{Computing the PSC Value of an Instance}
\label{sec:comp-apscI}

Computing the minimal possible $\alpha$-value that is achievable in an instance by \textit{any} committee is more challenging. We first show that the problem is NP-complete.%

\begin{restatable}{thm}{hardness}
Given an instance and a value $\alpha<1$, deciding whether $\alpha$-PSC is satisfiable is NP-complete.
\label{thm:hardness}
\end{restatable}
\begin{proof} 

Membership in NP follows from \Cref{thm:psc-W}, as any $W$ with $\apsc(W)\le \alpha$ witnesses the satisfiability of $\alpha$-PSC.

To show hardness, we reduce from 3-Hitting Set. Recall that in the 3-Hitting Set problem, we are given a set $S$, a collection of size-three subsets $S_1, ..., S_j$ of $S$ and an integer $h<j$, and the goal is to find a set $H \subseteq S$ of size $|H|=h$ such that $H \cap S_i \neq \emptyset$ for all $i \in [j]$. For each $i \in [j]$, let  $S_i=\{s^1_i, s^2_i, s^3_i\}$. For a given instance $(S, \{S_1, \dots, S_j\}, h)$, we construct a corresponding election instance %
as follows. 

Let $C = S \cup D$, where $D = \{d_1, \dots, d_j\}$ is a set of $j$ dummy candidates. 
For each $i \in [j]$, there are two voters $v_i^1,v_i^2$ with $A_i = S_i \cup \{d_i\}$ and preferences
\[
    v_i^1: \, d_i \succ s^1_i \succ s^2_i \succ s^3_i \quad \text{and} \quad
    v_i^2: \, d_i \succ s^2_i \succ s^3_i \succ s^1_i \, .
\]
Thus, $n=2j$. Finally, we let $k = h + j$ and $\alpha = \frac{k}{n}<1$. 

Since $\alpha \frac{n}{k}=1$, every individual voter on its own constitutes a $1_\alpha$-large solid coalition %
and, for each $i \in [j]$, the voters $v_i^1$ and $v_i^2$ together form a $2_\alpha$-large solid coalition over the prefix $\{d_i\} \cup S_i$. 
Hence, a committee satisfying $\alpha$-PSC must contain all dummy candidates and at least one candidate from each of the sets $S_1, \dots, S_j$. It follows that there is a committee of size $k$ that satisfies $\alpha$-PSC if and only if there exists a hitting set of size $h$.

This proof can easily be extended to any $\alpha < 1$ by either cloning the voters or adding new sets consisting only of a single otherwise unused element. 
\end{proof}

Thus, in order to compute PSC values in our experiments (see \Cref{sec:exp}), we employ integer linear programming (ILP).
The approach is similar to the one used in \Cref{sec:comp-apscW}: We compute the set $T$ of threshold values and then consider these values in non-increasing order. When considering $\alpha^\ell_{(N',C')}$, we add the constraint $|W \cap C'| \ge \ell$ to our ILP and check whether the resulting ILP is feasible. If yes, we consider the next threshold in $T$. If not, we have found the PSC value of the instance, as $\alpha^\ell_{(N',C')}$ is the largest value of $\alpha$ for which $\alpha$-PSC is not satisfiable. 

Formally, the ILP has a binary variable $x_c \in \{0,1\}$ for each candidate $c \in C$ and a constraint $\sum_{c \in C} x_c \le k$ ensuring that at most $k$ candidates are selected. Constraints of the form $|W \cap C'| \ge \ell$ can be encoded as $\sum_{c \in C'} x_c \ge \ell$.

We remark that this algorithm has similarities to the \textit{D'Hondt apportionment method} \citep{BaYo82a}. In the full version of this paper, we develop a description of this algorithm which gives rise to the idea of ``{apportionment for non-disjoint parties},'' which might be of independent interest.

\subsection{Quantifying Other Axioms} 
\label{sec:qOther}

Generalizing the quantification approach to local stability and EJR+ is straightforward. Similarly to PSC, we can replace each $\ell$-large group by an $\ell_\alpha$-large group leading to the following two definitions. 
\begin{definition}[$\alpha$-LS]
    A committee $W$ satisfies \emph{$\alpha$-local stability ($\alpha$-LS)} if there is no $1_\alpha$-large group $N' \subseteq N$ of voters and $c \notin W$ such that 
    $c \succ_i c'$ for all $i \in N'$ and~$c' \in W$.
\end{definition}
While $1$-LS might not be achievable, \citet{CLP+24a} recently showed that every instance admits a $9.8217$-LS committee.
For the definition of $\alpha$-EJR+, we let $\rank(i,c) = \lvert \{c' \in A_i: c' \succ_i c\}| + 1$ denote the rank that voter $i$ assigns to candidate $c$. If $c \notin A_i$, we let $\rank(i,c)= m$. 

\begin{definition}[$\alpha$-EJR+]
    A committee $W$ satisfies \emph{$\alpha$-EJR+} if there is no $\ell \in [k]$, $\ell_\alpha$-large group $N'\subseteq N$ of voters, unselected candidate $c \notin W$, and rank $r \in [m]$ such that
    \begin{enumerate}[(i)]
        \item $\rank(i,c) \le r $ for all $i \in N'$
        \item $\lvert \{c' \in C \colon \rank(i,c') \le r\} \cap W \rvert < \ell$ for all $i \in N'$.
    \end{enumerate}
\end{definition}

The definition of priceability %
is already parameterized, with a price of $p < \frac{n}{k}$ implying PSC \citep{BrPe23a}.
It is easy to generalize this implication to show that if the lowest possible price is $p$, the corresponding committee satisfies $p \frac{n}{k}$-PSC. Thus, we say that a committee satisfies \emph{$\alpha$-priceability} if the smallest price~$p$ for which the committee is priceable %
satisfies $p \le \alpha\frac{n}{k}$.

For all three notions, the minimal $\alpha$-value achieved by a given committee can be computed in polynomial time. 
For local stability and EJR+, it is sufficient to iterate over the unchosen candidates and compare the size of the coalitions that would want to deviate to these candidates. 
The optimal price for priceability can be computed via a linear program. 

We note that it is already NP-complete to decide whether any locally stable committee exists \cite{AEF+17a}. Further, the construction in the proof of \Cref{thm:hardness} also applies to both priceability and EJR+, showing that computing the minimal $\alpha$-value for these two measures is also NP-hard.

\section{Experimental Results}
\label{sec:exp}
To assess the measures defined in Section \ref{sec:qProp}, we conducted several experiments on the 1070 election instances from the dataset discussed in \Cref{sec: data}. We highlight some of our results in this section, mostly focusing on PSC; all remaining results can be found in the full version of this paper.

\begin{table}[tb]
\centering
\begin{tabular}{@{}lcccc@{}} 
\toprule 
  & M-STV & EAR  & SNTV  & seq-RCV \\ %
 \midrule
S-STV  & 108 (10.1\%) & 262 (24.5\%) & 277 (25.9\%) & 485 (45.3\%) \\ %
M-STV   & -- & 230 (21.5\%) & 333 (31.1\%) & 415 (38.8\%) \\ %
EAR  &  & -- & 452 (42.2\%) & 459 (42.9\%) \\ %
SNTV &  &  & -- & 599 (56.0\%) \\ %
\bottomrule
\end{tabular}%
    \caption{Number of instances on which the rules disagree.}
    \label{tab:rules-num-disagreement}
\end{table}

We considered the following voting rules:
Scottish STV (\textit{S-STV}),
Meek STV (\textit{M-STV}), 
EAR,
SNTV, and 
seq-RCV. 
\Cref{tab:rules-num-disagreement} shows how often these rules disagree with each other (i.e., choose different committees) on our data. We observe that S-STV and M-STV agree very frequently, but not in all elections. Further, both STV variants agree with SNTV in nearly $70\%$ of the elections, i.e., in most elections both STV variants simply select the $k$ candidates with the most first-place votes. This is slightly less for EAR, which agrees with SNTV in only $58\%$ of the elections. Further, seq-RCV seems to be the rule that is most different from the other rules, agreeing with SNTV in only $45\%$ of the cases.

\paragraph{Minimal Values}
For each instance, we computed the optimal $\alpha$-value for each measure (see \Cref{fig:min-psc-ejr-hist} for histograms of values for PSC and EJR+).
For all measures, the majority of minimal $\alpha$-values lie roughly in the range $0.4$ to $0.6$, the PSC values overall being somewhat lower than for the other measures (EJR+ and priceability in particular). Interestingly, we observe that while a $1$-LS or $1$-EJR+ committee is not guaranteed to exist in general, they always exist for the instances in our dataset. This observation is similar in spirit to the observation that Condorcet winners almost always exist in real-world elections \citep{McMc24a}. 

\begin{figure}[tb]
    \centering
    \includegraphics[width=0.9\linewidth]{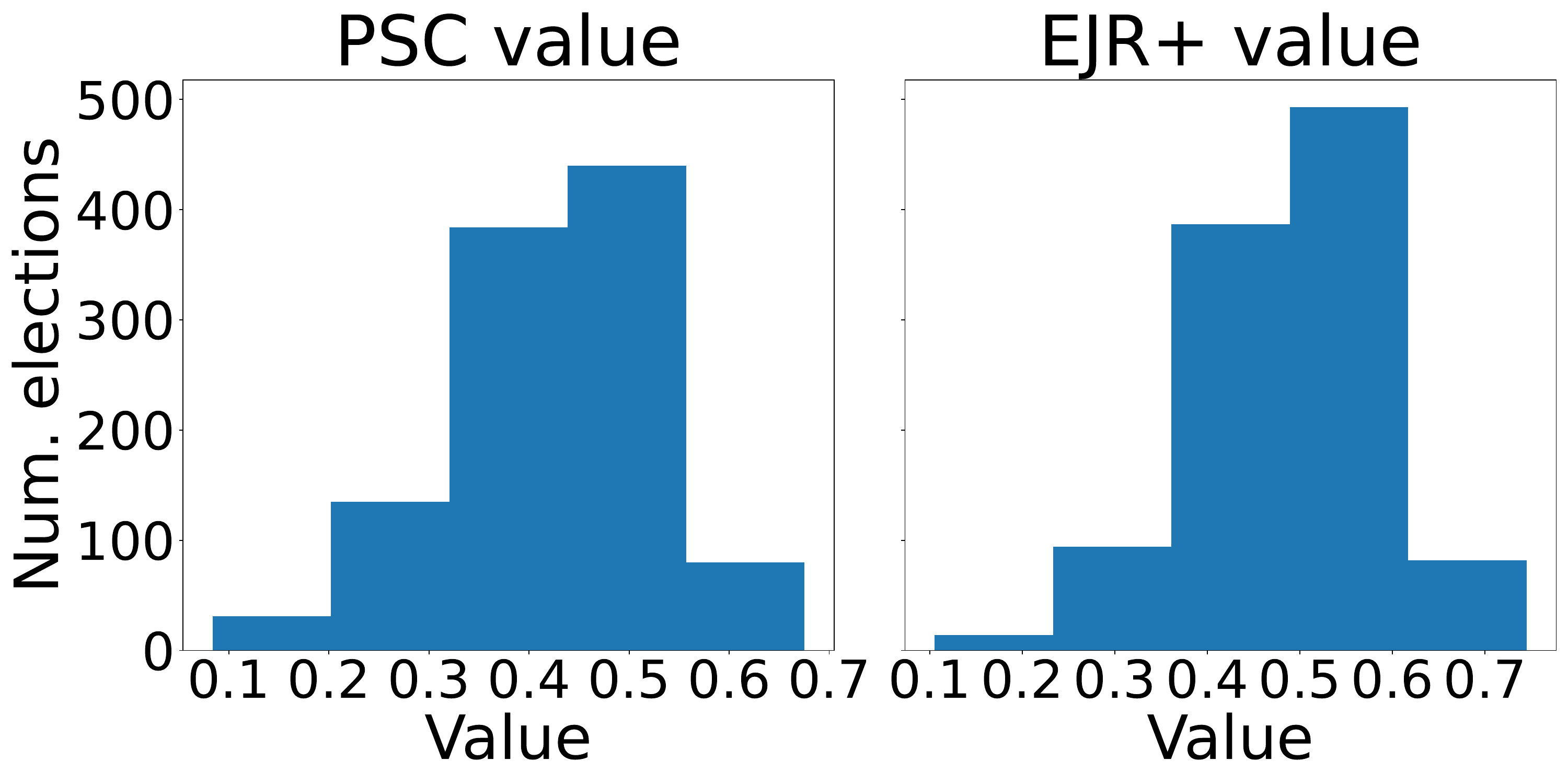}
    \caption{Histograms of PSC values and EJR+ values achievable in our elections, rounded to one decimal place.}
    \label{fig:min-psc-ejr-hist}
\end{figure}

\paragraph{Distance from Optimality}
For each voting rule, 
we counted 
\textit{(i)} how often the rule achieves the optimal $\alpha$-value and 
\textit{(ii)} the average distance between the committees chosen by the voting rule and the committees optimizing the $\alpha$-value.\footnote{For \textit{(ii)}, we define the distance between two committees as half of their symmetric difference.}
The results, presented in \Cref{tab:optimal-ptage-avg-dist}, reveal which voting rules are ``most aligned'' with each of the four measures. In particular, SNTV is most aligned with the PSC and LS measures, and the STV rules are most aligned with the EJR+ and priceability measures. The strong performance of SNTV can be considered surprising insofar as the rule does not satisfy any proportionality guarantees. A possible explanation for the good values achieved by SNTV (which outperforms EAR according to all four measures) can be found in the structure of our data: often, most of the constraints that a quantified proportionality axiom like $\alpha$-PSC imposes involve top-ranked candidates only, and SNTV\,---\,by definition\,---\,selects the candidates with the most first-place votes. 

\begin{table}[t]
    \centering
\begin{tabular}{@{}lcccccccc@{}}
\toprule
& \multicolumn{2}{c}{PSC} & \multicolumn{2}{c}{EJR+} & \multicolumn{2}{c}{Priceability} & \multicolumn{2}{c}{LS}\\
\cmidrule(r){2-3}\cmidrule(lr){4-5}\cmidrule(lr){6-7}\cmidrule(lr){8-9}
 & opt. & dist. & opt. & dist. & opt. & dist. & opt. & dist. \\
 \midrule
S-STV & 856 & 0.20 & 826 & 0.23 & \textbf{870} & \textbf{0.19} & 829 & 0.23 \\ %
M-STV & 819 & 0.24 & \textbf{842} &  \textbf{0.22} & 840 & 0.22 & 754 & 0.30  \\ %
EAR & 677 & 0.38 & 737 & 0.33 & 709 & 0.35 & 656 & 0.40 \\ %
SNTV & \textbf{901} & \textbf{0.16} & 752 & 0.30 & 832 & 0.23 & \textbf{935} & \textbf{0.13} \\ 
seq-RCV & 552 & 0.50 & 646 & 0.40 & 586 & 0.46 & 459 & 0.60 \\ 
\bottomrule
\end{tabular}%
    \caption{For each rule and each axiom, \textit{(i)} ``opt.'' refers to the number of instances for which the rule achieves the optimal $\alpha$-value and \textit{(ii)} ``dist.'' refers to the average distance between the outcome of the rule and the outcome with optimal $\alpha$-value (measured in terms of number of candidates that need to be exchanged). The best values in each column appear in bold.}
    \label{tab:optimal-ptage-avg-dist}
\end{table}

\begin{figure}[tb]
    \centering
    \includegraphics[width=\linewidth]{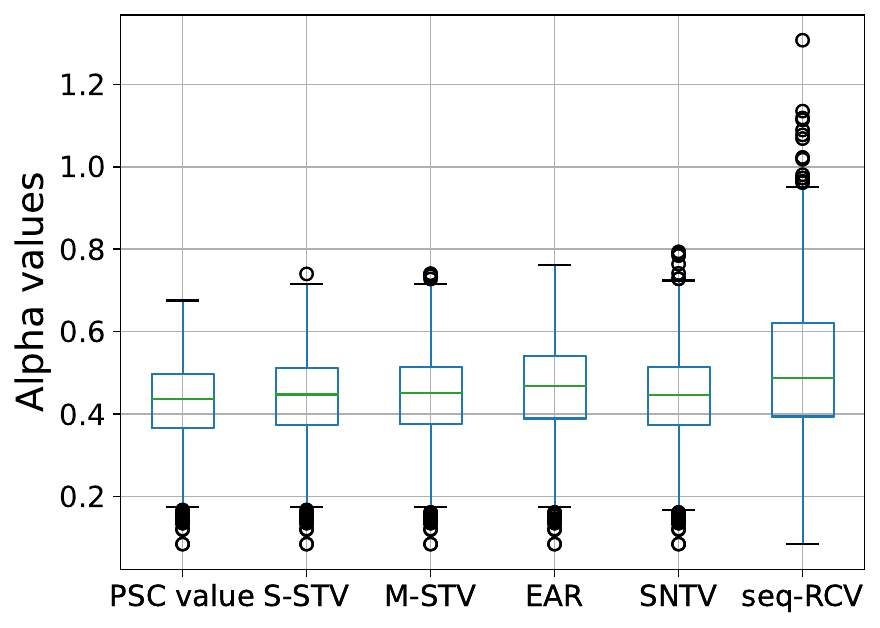}
    \caption{The PSC values achieved by voting rules, together with optimal PSC values (shown in the leftmost column).
    }
    \label{fig:alphas-psc-boxplot-rules}
\end{figure}

\paragraph{Values Achieved by Rules}
We furthermore computed the spread of $\alpha$-values achieved by different voting rules over the set of all instances and compared these values to the spread of optimal $\alpha$-values.  
For PSC values, the results are presented as a box plot in \Cref{fig:alphas-psc-boxplot-rules}.
Overall, all proportional rules\,---\,and the semi-proportional SNTV\,---\,perform similarly in terms of approximating optimal values, with the range of values for each of the rules coming close to those of the optimal values.\footnote{
The outlier value at $\alpha \approx 0.08$ 
stems from the 2012 election of North Lanarkshire, Ward 9, %
where $3$ out of $4$ candidates needed to be elected.
In this election, all rules and measures choose the same committee and the only unselected candidate is greatly unpopular.
}
Somewhat surprisingly, EAR\,---\,the rule satisfying the strongest proportionality axioms (see \Cref{sec:rules})\,---\,does slightly worse than the other proportional rules. (This is also apparent in \Cref{tab:optimal-ptage-avg-dist}.) 

Furthermore, SNTV does slightly better than the other rules w.r.t. PSC values (and the same is true for LS). A reason for that, as already discussed in the context of \Cref{tab:optimal-ptage-avg-dist}, is that the measures often require the $k$ most popular candidates to be chosen: on average, 57\% of the constraints corresponding to the optimal PSC value are over singleton sets of candidates, and thus correspond directly to first-place votes.

\paragraph{Pairwise Comparisons}
Finally, we considered pairwise comparisons of voting rules w.r.t. the $\alpha$-values they achieve. In these comparisons, we only consider instances on which the two rules under consideration output different committees. We focus on two comparisons w.r.t. PSC values: S-STV vs.\ seq-RCV and  
 S-STV vs.\ EAR (\Cref{fig:stv-rcv-ear-diff-side}). 
In both of these cases, S-STV does better in terms of PSC values. However, as one might expect, the overall difference in values between S-STV and the non-proportional seq-RCV is much more pronounced than the difference between S-STV and EAR. In particular, seq-RCV fails $1$-PSC
in 9 instances.

    \label{fig:stv-ear-diff}

\begin{figure}[t]
    \centering
    \includegraphics[width=\linewidth]{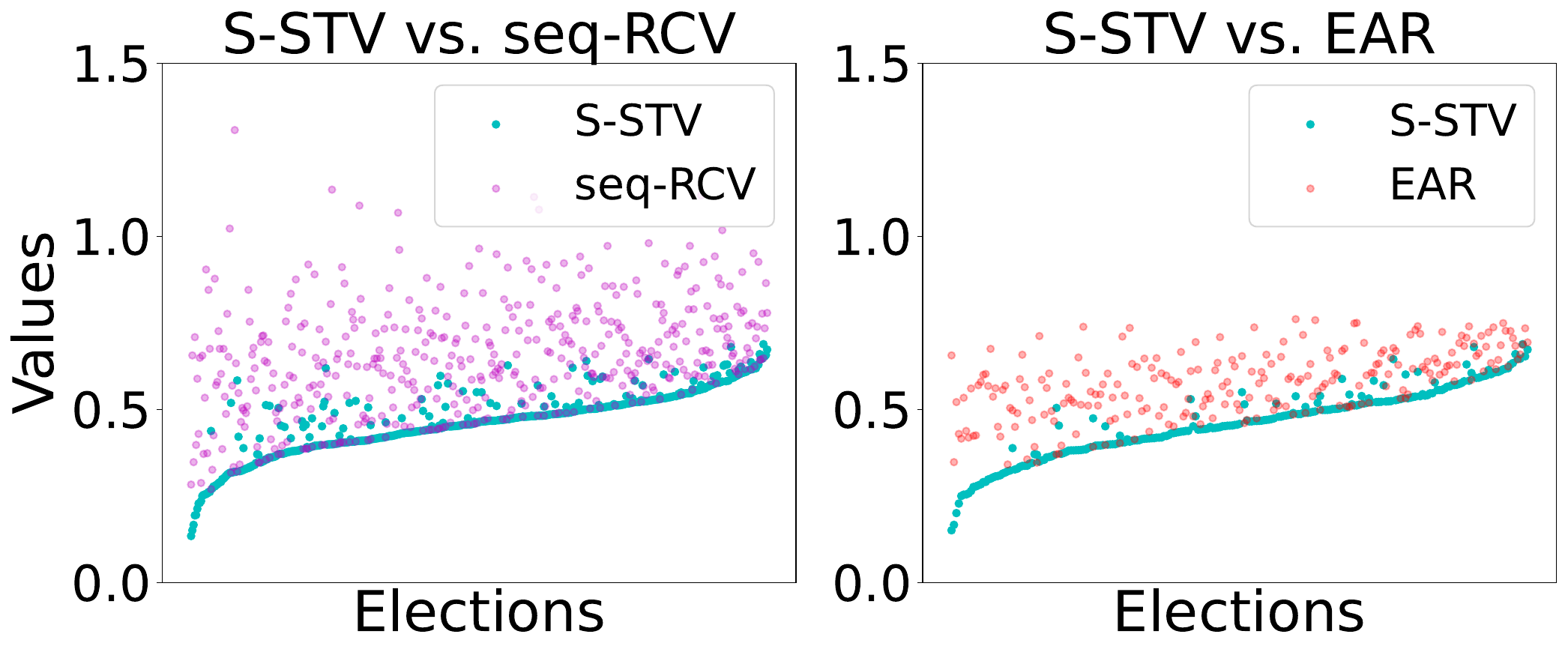}
    \caption{\textit{Left:} PSC values achieved by Scottish STV and seq-RCV for the 485 elections where the rules disagree. \textit{Right:} PSC values achieved by Scottish STV and EAR over the 262 elections where the rules disagree. Elections are ordered by increasing optimal PSC value.}
    \label{fig:stv-rcv-ear-diff-side}
\end{figure}

\section{The Effect of Ballot Truncation}
\label{sec:truncation}

To understand how ballot truncation affects our results, we created ballot data with complete rankings based on the Scottish data. To create this synthetic data, we employed an iterative process that is described in the full version of this paper. Basically, when extending partial ballots of length $r$ to length $r+1$, we consider the frequency of ballots of length at least $r+1$ which agree on the first $r$ entries. 

We then reran all experiments from \Cref{sec:exp} for the 1070 synthetic election instances with complete rankings. Overall, the results for completed instances are very similar to the results for the original (truncated) instances. For instance, in $21.8\%$ of these elections any committee of size $k$ is compatible with PSC (compared to $27.5\%$ in the truncated case; see \Cref{tab:complete}). This implies that the effect of ballot truncation is rather limited for the elections we study. This is a bit of a surprise, as these results suggest that the primary reason PSC has little discriminatory power in real-world elections is \textit{not} that voters truncate their ballots; rather, even if preferences are completed, voters do not form sufficiently large cohesive groups. 

In general, as expected, the achieved $\alpha$-values are slightly larger in the completed instances, with, for instance, SNTV also sometimes violating EJR+. Further, we observe that in most instances $\frac{k}{k+1}$ is a lower bound on the lowest possible priceability value, and that most priceability values achieved by the rules are clustered around that threshold. 

With completed preferences, seq-RCV violates PSC in 55 elections and achieves $\alpha$-values of up to $1.6$ for local stability. This suggests that non-proportional methods become even more noticeably non-proportional when preferences are not truncated.

\begin{table}
    \centering
\begin{tabular}{@{}lrrr@{}}
\toprule
 & $<25\%$\phantom{xx} & $<50\%$\phantom{xx} & $<100\%$\phantom{xx} \\ 
\cmidrule(r){2-2}\cmidrule(lr){3-3}\cmidrule(lr){4-4} %
PSC & 49 (4.6\%) & 305 (28.5\%) & 837 (78.2\%) \\ %
EJR+ & 81 (7.5\%) & 430 (40.2\%) & 1030 (96.3\%) \\ %
LS & 240 (22.4\%) & 749 (70.0\%) & 1066 (99.6\%)  \\ %
Priceability & 115 (10.7\%) & 503 (47.0\%) & 1036 (96.8\%)  \\ 
\bottomrule
\end{tabular}%
\caption{Analog of \Cref{tab:committees-axioms} for completed ballots.} 
    \label{tab:complete}
\end{table}

\section{Conclusion}

Given the absence of large cohesive groups of voters in real-world political elections, we proposed adaptations of several established proportionality axioms in which we loosen size constraints of cohesive groups, thereby creating new ways to quantify proportionality in practice.  Our results show that while delivering separations in theory, in practice these proportionality measures seem to behave similarly for proportional rules.  A majoritarian method like seq-RCV, on the other hand, performs poorly w.r.t. our measures. 
We also found that SNTV, a very simple rule without proportionality guarantees, performs well in practice in most cases. 
This study is a first attempt to grapple with the meaning of empirical proportionality using a large real-world dataset.

There are multiple ways to build upon our work. First, one could try to reconcile theory and practice by coming up with an axiomatic explanation for the performance of STV that goes beyond PSC. Second, our results motivate the search for proportionality axioms (or measures) that are better suited for assessing the real-world performance of voting rules. Finally, it would be interesting to obtain ballot data from some of the various other jurisdictions that use STV and to check whether those elections exhibit the same effects that we observed in the Scottish election dataset. 

\clearpage
\section*{Acknowledgments}
This work was partially supported by the Singapore Ministry of Education under grant number MOE-T2EP20221-0001.
\bibliographystyle{abbrvnat}
\bibliography{abb,algo,bibliography}

\begin{thebibliography}{48}
\providecommand{\natexlab}[1]{#1}
\providecommand{\url}[1]{\texttt{#1}}
\expandafter\ifx\csname urlstyle\endcsname\relax
  \providecommand{\doi}[1]{doi: #1}\else
  \providecommand{\doi}{doi: \begingroup \urlstyle{rm}\Url}\fi

\bibitem[Amy(2000)]{Amy00x}
D.~J. Amy.
\newblock \emph{Behind the Ballot Box: A Citizen's Guide to Voting Systems}.
\newblock Bloomsbury Academic, 2000.

\bibitem[Aziz and Lee(2020)]{AzLe20a}
H.~Aziz and B.~E. Lee.
\newblock The expanding approvals rule: improving proportional representation and monotonicity.
\newblock \emph{Social Choice and Welfare}, 54:\penalty0 1--45, 2020.

\bibitem[Aziz and Lee(2021)]{AzLe21a}
H.~Aziz and B.~E. Lee.
\newblock Proportionally representative participatory budgeting with ordinal preferences.
\newblock In \emph{Proceedings of the 35th AAAI Conference on Artificial Intelligence (AAAI)}, pages 5110--5118. AAAI Press, 2021.

\bibitem[Aziz and Lee(2022)]{AzLe22a}
H.~Aziz and B.~E. Lee.
\newblock A characterization of proportionally representative committees.
\newblock \emph{Games and Economic Behavior}, 133:\penalty0 248--255, 2022.

\bibitem[Aziz et~al.(2017{\natexlab{a}})Aziz, Brill, Conitzer, Elkind, Freeman, and Walsh]{ABC+16a}
H.~Aziz, M.~Brill, V.~Conitzer, E.~Elkind, R.~Freeman, and T.~Walsh.
\newblock Justified representation in approval-based committee voting.
\newblock \emph{Social Choice and Welfare}, 48\penalty0 (2):\penalty0 461--485, 2017{\natexlab{a}}.

\bibitem[Aziz et~al.(2017{\natexlab{b}})Aziz, Elkind, Faliszewski, Lackner, and Skowron]{AEF+17a}
H.~Aziz, E.~Elkind, P.~Faliszewski, M.~Lackner, and P.~Skowron.
\newblock The {C}ondorcet principle for multiwinner elections: from shortlisting to proportionality.
\newblock In \emph{Proceedings of the 26th International Joint Conference on Artificial Intelligence (IJCAI)}, pages 84 -- 90, 2017{\natexlab{b}}.

\bibitem[Aziz et~al.(2024)Aziz, Lee, {Morota Chu}, and Vollen]{ALM23a}
H.~Aziz, B.~E. Lee, S.~{Morota Chu}, and J.~Vollen.
\newblock Proportionally representative clustering.
\newblock In \emph{Proceedings of the 20th International Workshop on Internet and Network Economics (WINE)}, 2024.
\newblock Forthcoming.

\bibitem[Balinski and Young(1980)]{BaYo80a}
M.~L. Balinski and H.~P. Young.
\newblock The webster method of apportionment.
\newblock \emph{Proceedings of the National Academy of Sciences (PNAS)}, 77\penalty0 (1):\penalty0 1--4, 1980.

\bibitem[Balinski and Young(1982)]{BaYo82a}
M.~L. Balinski and H.~P. Young.
\newblock \emph{Fair Representation: {M}eeting the Ideal of One Man, One Vote}.
\newblock Yale University Press, 1982.

\bibitem[Boehmer et~al.(2023)Boehmer, Faliszewski, Janeczko, and Kaczmarczyk]{BFJK23a}
N.~Boehmer, P.~Faliszewski, {\L}.~Janeczko, and A.~Kaczmarczyk.
\newblock Robustness of participatory budgeting outcomes: Complexity and experiments.
\newblock In \emph{Proceedings of the 16th International Symposium on Algorithmic Game Theory (SAGT)}, pages 161--178, 2023.

\bibitem[Boehmer et~al.(2024{\natexlab{a}})Boehmer, Brill, Cevallos, Gehrlein, S{\'a}nchez-Fern{\'a}ndez, and Schmidt-Kraepelin]{BBC+24a}
N.~Boehmer, M.~Brill, A.~Cevallos, J.~Gehrlein, L.~S{\'a}nchez-Fern{\'a}ndez, and U.~Schmidt-Kraepelin.
\newblock Approval-based committee voting in practice: A case study of (over-)representation in the {P}olkadot blockchain.
\newblock In \emph{Proceedings of the 38th AAAI Conference on Artificial Intelligence (AAAI)}, pages 9519--9527, 2024{\natexlab{a}}.

\bibitem[Boehmer et~al.(2024{\natexlab{b}})Boehmer, Faliszewski, Janeczko, Kaczmarczyk, Lisowski, Pierczy{\'n}ski, Rey, Stolicki, Szufa, and W{a}s]{BFJ+24x}
N.~Boehmer, P.~Faliszewski, {\L}.~Janeczko, A.~Kaczmarczyk, G.~Lisowski, G.~Pierczy{\'n}ski, S.~Rey, D.~Stolicki, S.~Szufa, and T.~W{a}s.
\newblock Guide to numerical experiments on elections in computational social choice.
\newblock In \emph{Proceedings of the 33rd International Joint Conference on Artificial Intelligence (IJCAI)}, pages 7962--7970, 2024{\natexlab{b}}.

\bibitem[Boehmer et~al.(2024{\natexlab{c}})Boehmer, Faliszewski, Janeczko, Peters, Pierczy{\'n}ski, Schierreich, Skowron, and Szufa]{BFJ+24b}
N.~Boehmer, P.~Faliszewski, {\L}.~Janeczko, D.~Peters, G.~Pierczy{\'n}ski, {\v S}.~Schierreich, P.~Skowron, and S.~Szufa.
\newblock Evaluation of project performance in participatory budgeting.
\newblock In \emph{Proceedings of the 33rd International Joint Conference on Artificial Intelligence (IJCAI)}, pages 2678--2686, 2024{\natexlab{c}}.

\bibitem[Brill and Peters(2023)]{BrPe23a}
M.~Brill and J.~Peters.
\newblock Robust and verifiable proportionality axioms for multiwinner voting.
\newblock In \emph{Proceedings of the 24th ACM Conference on Economics and Computation (ACM-EC)}, page 301. ACM Press, 2023.
\newblock Full version arXiv:2302.01989 [cs.GT].

\bibitem[Brill et~al.(2018)Brill, Laslier, and Skowron]{BLS18a}
M.~Brill, J.-F. Laslier, and P.~Skowron.
\newblock Multiwinner approval rules as apportionment methods.
\newblock \emph{Journal of Theoretical Politics}, 30\penalty0 (3):\penalty0 358--382, 2018.

\bibitem[Brill et~al.(2023)Brill, Forster, Lackner, Maly, and Peters]{BFL+23a}
M.~Brill, S.~Forster, M.~Lackner, J.~Maly, and J.~Peters.
\newblock Proportionality in approval-based participatory budgeting.
\newblock In \emph{Proceedings of the 37th AAAI Conference on Artificial Intelligence (AAAI)}, pages 5524--5531. AAAI Press, 2023.

\bibitem[Burnett and Kogan(2015)]{BK15}
C.~Burnett and V.~Kogan.
\newblock Ballot (and voter) ``exhaustion''' under instant runoff voting: An examination of four ranked-choice elections.
\newblock \emph{Electoral Studies}, 37:\penalty0 41--49, 2015.

\bibitem[Caragiannis et~al.(2024)Caragiannis, Micha, and Shah]{CMS24a}
I.~Caragiannis, E.~Micha, and N.~Shah.
\newblock Proportional fairness in non-centroid clustering.
\newblock In \emph{Proceedings of the 38th Conference on Neural Information Processing Systems (NeurIPS)}, 2024.
\newblock Forthcoming.

\bibitem[Charikar et~al.(2025)Charikar, Lassota, Ramakrishnan, Vetta, and Wang]{CLP+24a}
M.~Charikar, A.~Lassota, P.~Ramakrishnan, A.~Vetta, and K.~Wang.
\newblock Six candidates suffice to win a voter majority.
\newblock In \emph{Proceedings of the 57th Annual ACM Symposium on Theory of Computing (STOC)}, 2025.
\newblock Forthcoming.

\bibitem[Chen et~al.(2019)Chen, Fain, Lyu, and Munagala]{CFL+19a}
X.~Chen, B.~Fain, L.~Lyu, and K.~Munagala.
\newblock Proportionally fair clustering.
\newblock In \emph{Proceedings of the 36th International Conference on Machine Learning (ICML)}, pages 1032--1041, 2019.

\bibitem[Delemazure and Peters(2024)]{DePe24a}
T.~Delemazure and D.~Peters.
\newblock Generalizing instant runoff voting to allow indifferences.
\newblock In \emph{Proceedings of the 25th ACM Conference on Economics and Computation (ACM-EC)}, page~50. ACM Press, 2024.

\bibitem[Dummett(1984)]{Dumm84a}
M.~Dummett.
\newblock \emph{Voting Procedures}.
\newblock Oxford University Press, 1984.

\bibitem[Ebadian and Micha(2025)]{EbMi23a}
S.~Ebadian and E.~Micha.
\newblock Boosting sortition via proportional representation.
\newblock In \emph{Proceedings of the 24th International Conference on Autonomous Agents and Multiagent Systems (AAMAS)}, 2025.
\newblock Forthcoming.

\bibitem[Ebadian et~al.(2024)Ebadian, Kahng, Peters, and Shah]{EKPS24a}
S.~Ebadian, A.~Kahng, D.~Peters, and N.~Shah.
\newblock Optimized distortion and proportional fairness in voting.
\newblock \emph{ACM Transactions on Economics and Computation}, 12\penalty0 (1):\penalty0 1---39, 2024.

\bibitem[Elkind et~al.(2015)Elkind, Lang, and Saffidine]{ELS15a}
E.~Elkind, J.~Lang, and A.~Saffidine.
\newblock Condorcet winning sets.
\newblock \emph{Social Choice and Welfare}, 44\penalty0 (3):\penalty0 493--517, 2015.

\bibitem[Elkind et~al.(2017)Elkind, Faliszewski, Laslier, Skowron, Slinko, and Talmon]{EFL+17a}
E.~Elkind, P.~Faliszewski, J.-F. Laslier, P.~Skowron, A.~Slinko, and N.~Talmon.
\newblock What do multiwinner voting rules do? {A}n experiment over the two-dimensional {E}uclidean domain.
\newblock In \emph{Proceedings of the 31st AAAI Conference on Artificial Intelligence (AAAI)}, pages 494--501. AAAI Press, 2017.

\bibitem[Faliszewski et~al.(2023)Faliszewski, Fils, Peters, Pierczy{\'n}ski, Skowron, Stolicki, Szufa, and Talmon]{FFP+23a}
P.~Faliszewski, J.~Fils, D.~Peters, G.~Pierczy{\'n}ski, P.~Skowron, D.~Stolicki, S.~Szufa, and N.~Talmon.
\newblock Participatory budgeting: Data, tools, and analysis participatory budgeting: Data, tools, and analysis.
\newblock In \emph{Proceedings of the 32nd International Joint Conference on Artificial Intelligence (IJCAI)}, pages 2667--2674, 2023.

\bibitem[Fish et~al.(2024)Fish, G{\"o}lz, Parkes, Procaccia, Rusak, Shapira, and W{\"u}thrich]{FGP+23a}
S.~Fish, P.~G{\"o}lz, D.~C. Parkes, A.~D. Procaccia, G.~Rusak, I.~Shapira, and M.~W{\"u}thrich.
\newblock Generative social choice.
\newblock In \emph{Proceedings of the 25th ACM Conference on Economics and Computation (ACM-EC)}, page 985. ACM Press, 2024.

\bibitem[Graham-Squire and McCune(2023)]{GSM23}
A.~Graham-Squire and D.~McCune.
\newblock An examination of ranked-choice voting in the united states, 2004–2022.
\newblock \emph{Representation}, pages 1--19, 2023.

\bibitem[Hill et~al.(1987)Hill, Wichmann, and Woodall]{HWW87x}
I.~D. Hill, B.~A. Wichmann, and D.~R. Woodall.
\newblock Algorithm 123: Single transferable vote by meek's method.
\newblock \emph{The Computer Journal}, 30\penalty0 (3), 1987.

\bibitem[Hoffman et~al.(2024)Hoffman, Kauba, Reidy, and Weighill]{HKRW21a}
C.~Hoffman, J.~Kauba, J.~Reidy, and T.~Weighill.
\newblock Proportionality in multi-winner {RCV} elections: A simulation study with ballot truncation.
\newblock \emph{Communications in Statistics - Simulation and Computation}, 2024.
\newblock Forthcoming.

\bibitem[Janson(2018)]{Jans18a}
S.~Janson.
\newblock Thresholds quantifying proportionality criteria for election methods.
\newblock Technical report, arXiv:1810.06377 [cs.GT], 2018.

\bibitem[Jiang et~al.(2020)Jiang, Munagala, and Wang]{JMW20a}
Z.~Jiang, K.~Munagala, and K.~Wang.
\newblock Approximately stable committee selection.
\newblock In \emph{Proceedings of the 52nd Annual ACM SIGACT Symposium on Theory of Computing (STOC)}, pages 463--472. ACM, 2020.

\bibitem[Kalayci et~al.(2024)Kalayci, Kempe, and Kher]{KKK24a}
Y.~H. Kalayci, D.~Kempe, and V.~Kher.
\newblock Proportional representation in metric spaces and low-distortion committee selection.
\newblock In \emph{Proceedings of the 38th AAAI Conference on Artificial Intelligence (AAAI)}, pages 9815--9823. {AAAI} Press, 2024.

\bibitem[Kellerhals and Peters(2024)]{KePe23a}
L.~Kellerhals and J.~Peters.
\newblock Proportional fairness in clustering: A social choice perspective.
\newblock In \emph{Proceedings of the 38th Conference on Neural Information Processing Systems (NeurIPS)}, 2024.

\bibitem[Kilgour et~al.(2020)Kilgour, Gr{\'e}goire, and Foley]{KGF20}
D.~Kilgour, J.~Gr{\'e}goire, and A.~Foley.
\newblock The prevalence and consequences of ballot truncation in ranked-choice elections.
\newblock \emph{Public Choice}, 184:\penalty0 197--218, 2020.

\bibitem[Lackner and Skowron(2022)]{LaSk22a}
M.~Lackner and P.~Skowron.
\newblock \emph{Multi-Winner Voting with Approval Preferences}.
\newblock Springer, 2022.

\bibitem[Marsh and Plescia(2016)]{MaPl16a}
M.~Marsh and C.~Plescia.
\newblock Split-ticket voting in an {STV} system: choice in a non-strategic context.
\newblock \emph{Irish Political Studies}, 31\penalty0 (2):\penalty0 163--181, 2016.

\bibitem[McCormick(2007)]{Cor07a}
S.~T. McCormick.
\newblock \emph{Handbook on Discrete Optimization}, chapter~7.
\newblock Elsevier, 2007.

\bibitem[McCune and Graham-Squire(2024)]{McGr23a}
D.~McCune and A.~Graham-Squire.
\newblock Monotonicity anomalies in scottish local government elections.
\newblock \emph{Social Choice and Welfare}, 63\penalty0 (1):\penalty0 69--101, 2024.

\bibitem[McCune and McCune(2024)]{McMc24a}
D.~McCune and L.~McCune.
\newblock Does the choice of preferential voting method matter? an empirical study using ranked choice elections in the united states.
\newblock \emph{Representation}, 60\penalty0 (1):\penalty0 1--16, 2024.

\bibitem[McCune et~al.(2024)McCune, Martin, Latina, and Simms]{MMLS24a}
D.~McCune, E.~Martin, G.~Latina, and K.~Simms.
\newblock A comparison of sequential ranked-choice voting and single transferable vote.
\newblock \emph{Journal of Computational Social Science}, 7\penalty0 (1):\penalty0 643--670, 2024.

\bibitem[Peters and Skowron(2020)]{PeSk20a}
D.~Peters and P.~Skowron.
\newblock Proportionality and the limits of welfarism.
\newblock In \emph{Proceedings of the 21st ACM Conference on Economics and Computation (ACM-EC)}, pages 793--794. ACM Press, 2020.

\bibitem[Peters et~al.(2021)Peters, Pierczy{\'n}ski, and Skowron]{PPS21a}
D.~Peters, G.~Pierczy{\'n}ski, and P.~Skowron.
\newblock Proportional participatory budgeting with additive utilities.
\newblock In \emph{Advances in Neural Information Processing Systems (NeurIPS)}, volume~34, pages 12726--12737, 2021.

\bibitem[Skowron(2021)]{Skow21a}
P.~Skowron.
\newblock Proportionality degree of multiwinner rules.
\newblock In \emph{Proceedings of the 22nd ACM Conference on Economics and Computation (ACM-EC)}, pages 820--840. ACM Press, 2021.

\bibitem[Tideman(1995)]{Tide95a}
N.~Tideman.
\newblock The single transferable vote.
\newblock \emph{Journal of Economic Perspectives}, 9\penalty0 (1):\penalty0 27--38, 1995.

\bibitem[Tideman(2006)]{Tide06a}
N.~Tideman.
\newblock \emph{Collective Decisions And Voting: {T}he Potential for Public Choice}.
\newblock Ashgate, 2006.

\bibitem[Tomlinson et~al.(2023)Tomlinson, Ugander, and Kleinberg]{TUK23a}
K.~Tomlinson, J.~Ugander, and J.~Kleinberg.
\newblock Ballot length in instant runoff voting.
\newblock In \emph{Proceedings of the 37th AAAI Conference on Artificial Intelligence (AAAI)}, pages 5841--5849. {AAAI} Press, 2023.

\end{thebibliography}

\newpage

\appendix

\section*{Technical Appendix}

The appendix is structured as follows. \Cref{app:data-ex-1} lists all solid coalitions of the election instance from \Cref{ex:mloth20122}. In \Cref{app:apportionment} we discuss connections between the algorithm described in \Cref{sec:comp-apscI} and apportionment methods. \Cref{app:experiments} provides more details on the experiments discussed in \Cref{sec:exp}. Finally, \Cref{app:completed} is dedicated to experiments on completed ballot data (see \Cref{sec:truncation}).

\section{Additional Data for \Cref{ex:mloth20122}}
\label{app:data-ex-1}

Consider again the instance discussed in \Cref{ex:mloth20122}. Here, we use capital letters A, B, C, ... to denote candidates; the mapping of candidate names to letters is as follows. 

\begin{table}[h!]
\centering
\begin{tabular}{llrr}
\toprule
\textbf{Candidate} & \textbf{Party} & \textbf{Votes} & \textbf{Letter} \\
\midrule
D. Milligan (\textbf{DM})& Labour & 1,574 & E\\
L. Milliken (\textbf{LM})& Labour & 525 & F\\
J. Aitchison (\textbf{JA}) & Independent & 382 & A \\
B. Constable (\textbf{BC}) & SNP & 1,257 & C\\
T. Munro (\textbf{TM}) & SNP & 358 & G\\
I. Baxter (\textbf{IB}) & Greens & 671 & B\\
E. Cummings (\textbf{EC}) & Conservative & 365 & D\\
\bottomrule
\end{tabular}
\end{table}

The following list contains all 125 maximal solid coalitions for this instance.
Each solid coalition $(N',C')$ appears in the format ($C'$:~$|N'|$), and the list is ordered by the size $|N'|$ of coalitions. 

\medskip

(EF:~1624), (E:~1574), (CG:~1277), (C:~1257), (B:~671), (BCG:~554), (F:~525), (ABCDEFG:~460), (BEF:~405), (A:~382), (D:~365), (G:~358), (AEF:~345), (CEFG:~292), (BC:~239), (AE:~228), (BD:~228), (CEF:~216), (CEG:~212), (ABCEFG:~200), (BCEFG:~197), (CE:~167), (BE:~159), (AB:~158), (BCDEFG:~156), (EFG:~140), (ACG:~132), (DEF:~117), (CFG:~107), (ABCG:~107), (ABEF:~107), (BF:~102), (ABD:~101), (CDG:~82), (AC:~79), (ACEFG:~78), (BCDG:~77), (AD:~73), (ABDEF:~72), (ABE:~68), (BDEF:~65), (BCE:~64), (CDEFG:~64), (BCEG:~63), (DE:~62), (EG:~62), (ABC:~61), (ABCDG:~61), (BG:~59), (ABCDEF:~53), (BCEF:~49), (ABF:~48), (BCD:~47), (CF:~45), (AF:~44), (ABCEF:~42), (ABCDEG:~42), (CD:~40), (ACE:~40), (ACDEFG:~40), (BDE:~39), (ABCEG:~35), (ABCDFG:~34), (BCFG:~33), (FG:~29), (BCF:~29), (BDF:~29), (DF:~28), (ACEF:~28), (ACEG:~28), (BEFG:~27), (ABCFG:~25), (BCDEF:~25), (ABEFG:~23), (ABDE:~22), (AG:~21), (AEFG:~21), (BCDEG:~21), (ABCE:~20), (ADEF:~20), (ABG:~18), (BDG:~18), (BCDFG:~17), (BEG:~16), (BFG:~15), (ACDG:~15), (AEG:~14), (\mbox{ABDEFG}:~14), (ADE:~13), (ACFG:~13), (DG:~12), (ABCD:~12), (CDEF:~12), (ACD:~10), (ADF:~9), (ABDF:~9), (ABFG:~9), (CDEG:~9), (ABCDE:~8), (ACF:~7), (CDF:~7), (BCDE:~7), (CDFG:~7), (DEFG:~7), (ADG:~6), (CDE:~6), (ABEG:~5), (ACDEF:~5), (BDEFG:~5), (AFG:~4), (DEG:~4), (ABCF:~4), (BCDF:~4), (ABCDF:~4), (ACDEG:~4), (ACDFG:~4), (ADEFG:~4), (DFG:~3), (ACDE:~3), (ABDG:~2), (ACDF:~2), (ADEG:~2), (BDFG:~2), (BDEG:~1).

\section{An Apportionment Perspective}
\label{app:apportionment}

The algorithm used to compute the PSC value of an instance (see \Cref{sec:comp-apscI}) has similarities to the \textit{D'Hondt apportionment method} \citep{BaYo82a}. Apportionment methods distribute parliamentary seats among parties based on the parties' vote counts in a party-list election. In particular, the D'Hondt method can be computed by (1) constructing a table that contains the vote counts of all parties, the vote counts divided by 2, the vote counts divides by 3, and so on; and (2) iteratively assigning a seat to the party corresponding to the next-highest number in the table, until all $k$ seats have been assigned.\footnote{For example, see \citet[page~362]{BLS18a}.}

The ILP-based algorithm in \Cref{sec:comp-apscI} can be phrased in a similar way, based on the observation that the threshold values
$\alpha^\ell_{(N',C')} = \frac{|N'|}{n} \cdot \frac{k}{\ell}$
are proportional to $\frac{|N'|}{\ell}$. Therefore, we could interpret the maximal solid coalitions as `parties' and construct a table of the parties' sizes, 
the parties' sizes divided by 2, and so on. In particular, the numbers corresponding to a solid coalition $(N',C')$ are $|N'|, \frac{|N'|}2, \frac{|N'|}3, \dots, \frac{|N'|}{|C'|}$.
Then, we can iterate over the numbers in this table in non-increasing order, just like the D'Hondt method would do. 

There are two main differences between the algorithm from \Cref{sec:comp-apscI} and D'Hondt's method. 
First, our algorithm allocates ``representation guarantees'' instead of seats: When considering  $\frac{|N'|}{\ell}$ corresponding to solid coalition $(N',C')$, the constraint $|W \cap C'| \ge \ell$ is added. 
Second, we may give out more than $k$ representation guarantees: 
Since the candidate sets of different solid coalitions are not necessarily disjoint, candidates can satisfy more than one solid coalition at once. We only stop allocating further representation guarantees if doing so would lead to an infeasible collection of guarantees (as determined by our ILP).

This alternative description of the algorithm for computing the PSC value of an instance gives rise to the idea of \textit{apportionment for non-disjoint parties}. Interestingly, any divisor method can be employed instead of the D'Hondt method. For example, the \textit{Sainte-Laguë} method (aka Webster), which satisfies attractive properties \citep{BaYo80a}, uses divisors $1,3,5,7,\dots$ instead of $1,2,3,4,\dots$.

\section{Detailed Analysis of Experiments}
\label{app:experiments}

\subsection{Agreement of Rules}
\label{app:exp-rules}
To complement the results on how often rules disagree, we calculated the average distance between each pair of rules. The results are given in \Cref{tab:rules-avgdist}. The pairwise distance between rules is low overall. As we have seen in \Cref{tab:rules-num-disagreement}, except for SNTV and seq-RCV, each pair agrees in a majority of cases. Moreover, the committees returned by the different rules often differ by only one candidate when they disagree. 

\begin{table}
    \centering
\begin{tabular}{@{}lccccc@{}} 
\toprule 
 & S-STV & M-STV & EAR  & SNTV  & seq-RCV \\ %
 \midrule
S-STV & 0 & 0.10 & 0.25 & 0.26 & 0.46 \\ %
M-STV & - & 0 & 0.22 & 0.31 & 0.39  \\ %
EAR & - & - & 0 & 0.44 & 0.44 \\ %
SNTV & - & - & - & 0 & 0.58  \\ %
seq-RCV & - & - & - & - & 0 \\
\bottomrule
\end{tabular}
    \caption{Average distance between pairs of rules.}
    \label{tab:rules-avgdist}
\end{table}

\subsection{Alignment of Measures}
In order to check how closely aligned our four measures are, we calculated, for each election, the optimal committee w.r.t. each measure (i.e., the committee minimizing $\alpha$). 
\Cref{tab:ax-avgdists} shows how much the committees optimizing different measures differ from each other. The difference between the committees optimizing the LS value and committees optimizing other measures is overall higher than the difference between the other measures. One possible explanation for this is that the LS measure is mostly restricted to only consider first-place candidates, while the other measures impose constraints more broadly. 

\begin{table}
    \centering
\begin{tabular}{@{}lccc@{}} 
\toprule 
 & {EJR+} & Priceability  & LS\\ %
 \midrule
PSC  & 193 (0.18) & \phantom{1}91 (0.08) & 254 (0.24) \\ %
EJR+  & -- & 137 (0.12) & 346 (0.33) \\ %
Priceability  &  & -- & 275 (0.26) \\ %

\bottomrule
\end{tabular}
    \caption{Number of elections where measure-minimizing committees disagree and average distance between those committees (in parenthesis).
}
    \label{tab:ax-avgdists}
\end{table}

\subsection{Local Stability and Priceability Values}
\label{app:exp-values}
The optimal LS and priceability values are given in \Cref{fig:min-ls-price-hist}. As previously mentioned, since no LS values exceeds 1, there exists a locally stable committee in each election in our data set. 

\begin{figure}
    \centering
    \includegraphics[width=\linewidth]{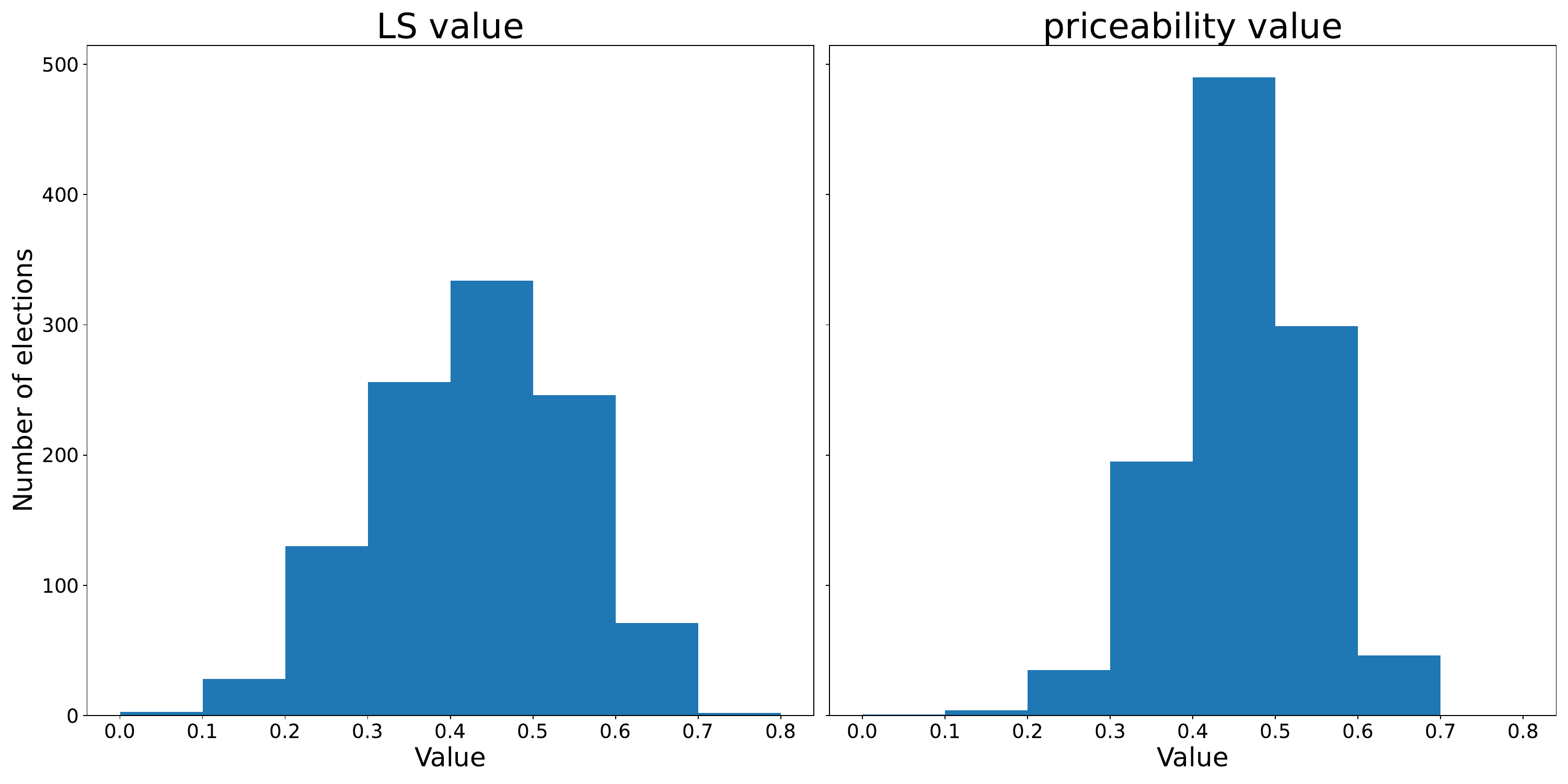}
    \caption{Histograms of LS and priceability values achievable in our elections, rounded to one decimal place.}
    \label{fig:min-ls-price-hist}
\end{figure}

\subsection{Values Achieved by Rules}
\label{app:exp-boxplots}
The EJR+, LS, and priceability values achieved by rules are given in \Cref{fig:alphas-ejr-boxplot-rules} to \Cref{fig:alphas-price-boxplot-rules}.

\begin{figure}
    \centering
    \includegraphics[width=\linewidth]{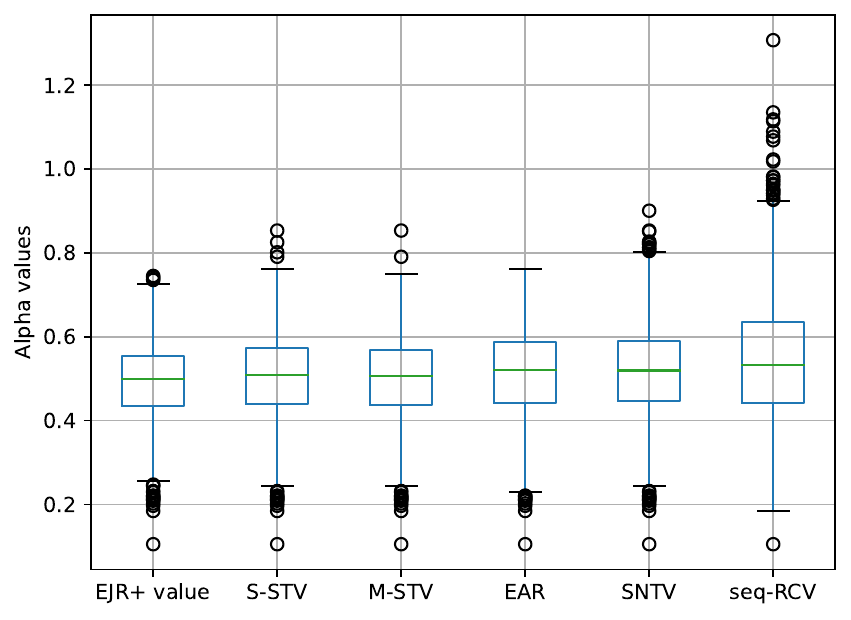}
    \caption{The EJR+ values achieved by voting rules}
    \label{fig:alphas-ejr-boxplot-rules}
\end{figure}

\begin{figure}
    \centering
    \includegraphics[width=\linewidth]{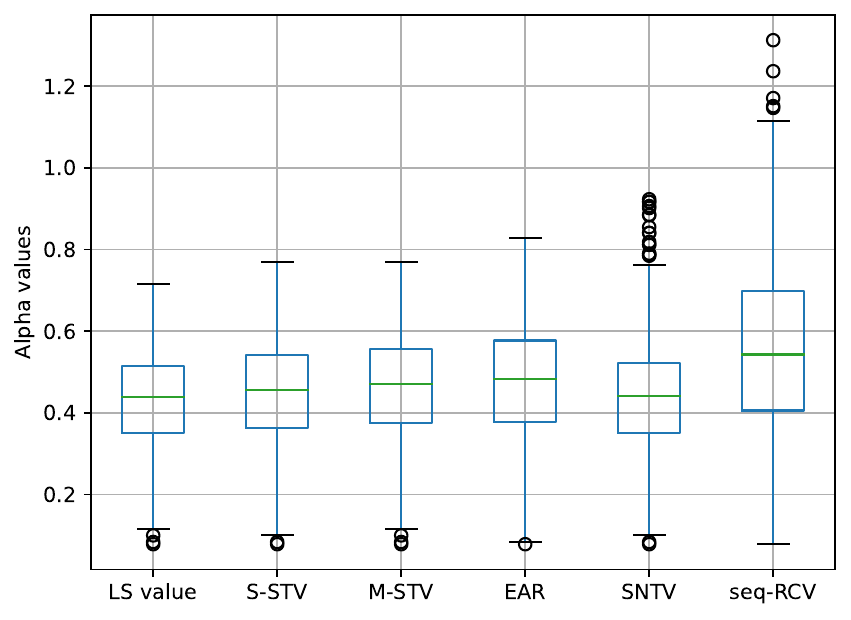}
    \caption{The LS values achieved by voting rules}
    \label{fig:alphas-ls-boxplot-rules}
\end{figure}

\begin{figure}
    \centering
    \includegraphics[width=\linewidth]{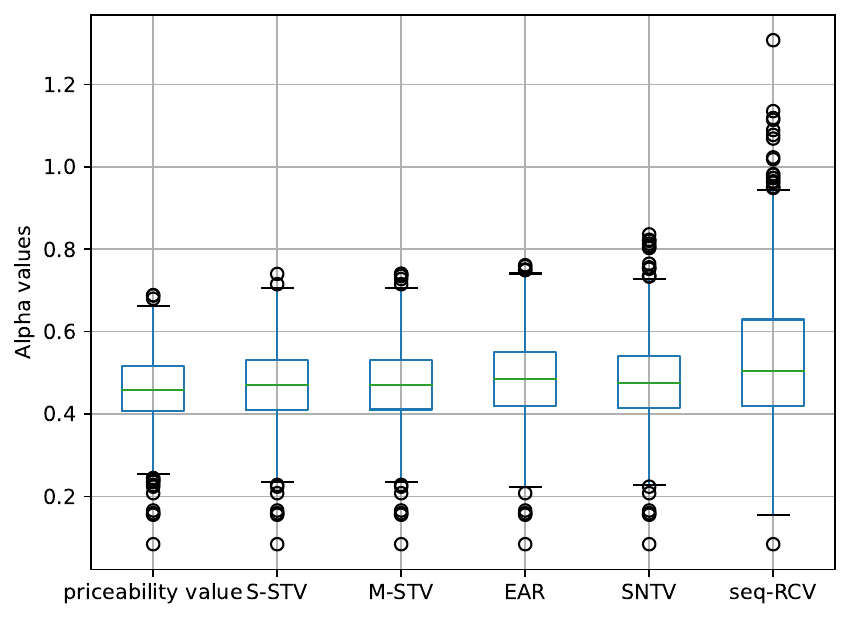}
    \caption{The priceability values achieved by voting rules}
    \label{fig:alphas-price-boxplot-rules}
\end{figure}

\clearpage
\section{Experimental Results for Completed Ballots}
\label{app:completed}
This section contains the results of experiments with completed ballot data, generated based on the Scottish data set. We first describe the method that we used to complete the data in Section~\ref{app:compgen}. As mentioned in Section 6, completing the data had limited effect on the experiments we performed. In Section~\ref{app:compex} we give an example to illustrate why the method described in Section~\ref{app:compgen} may not necessarily yield elections with (many) more sufficiently large solid coalitions. We give the results of the experiments on completed data in Sections~\ref{app:comp-hists-and-agreement} to \ref{app:pw}.

\subsection{Generating Completed Ballots}
\label{app:compgen}
We used the following method to create this synthetic data. The process involves a series of steps, each time extending partial ballots of length $r$ to length $r+1$ based on the probability distribution of  ballots which agree on the first $r$ entries. 

To be more concrete, suppose there are 10 ballots of the form $ABC$. To increase these ballots by length 1, we consider all ballots of the form $ABC*$ which have length at least 4. Suppose there are   38 such ballots as shown in Table \ref{completion}.  We extend the ballots of the form $ABC$ proportionally, (column 3) and round to a whole number using Hamilton's apportionment method (column 4). 

\begin{table}[tbh]
\scalebox{0.975}{
\begin{tabular}{lccc}
\toprule
Ballot&Number & Prop.  & Num. ballots\\
\midrule
ABCD & 9 &$\phantom{1}(9/38)\cdot 10 =2.368$ & 2 \\
ABCE & 12 & $(12/38)\cdot 10 = 3.158$& 3\\
ABCF& 17& $(17/38)\cdot 10 =4.474 $& 5\\
\midrule
Total & 38 & 10 & 10\\
\bottomrule
\end{tabular}
}
\caption{Extending ballots of length 3 to length 4.}
\label{completion}

\end{table}

In theory, this process can be iterated to create complete ballots, assuming a sufficient set of complete ballots. However, in practice it does not make sense to keep extending preferences on a given ballot if the number of ballots of length at least $r+1$ is not sufficiently large in comparison to the number of ballots of length $r$. Thus, each ballot was extended until either the ballot contained a complete ordering of the candidates, or until the number of ballots of length $r+1$ was less than  10\% of the number of ballots of length~$r$. 

If after this process, a ballot was not complete, we then completed the ballot by choosing between the candidates not yet listed on the ballot uniformly at random.

\subsection{Example Instance}
\label{app:compex}
As an example of why the number of large cohesive groups in elections does not increase substantially when moving from truncated data to data with full preferences, consider the 2022 election of Glasgow, Ward 20 (see \url{https://en.wikipedia.org/wiki/2022_Glasgow_City_Council_election#Baillieston}). %
In this election with $k = 3$, the two Labour candidates in the real instance with truncated ballots have $38.4\%$ of the first-place votes and form a solid coalition consisting of $0.313\%$ of the electorate (and thus too small to deserve representation under PSC). Furthermore, out of the voters that rank either of the Labour candidates first, $2.75\%$ only vote for this one candidate. Our completion method assigns the other Labour candidate the second place in most of these ballots, but not all of them, yielding a new solid coalition over these candidates that consists of $32.4\%$ of the electorate\,---\,barely below the $\frac{n}3$ threshold imposed by PSC. 

\subsection{Minimal Values}
\label{app:comp-hists-and-agreement}
We computed the optimal $\alpha$-value for each measure on the completed data. The histograms of encountered values is shown in \Cref{fig:hist-complete}. The minimal $\alpha$-values become slightly higher for all measures. The difference in values is most pronounced for EJR+ and priceability. In particular, the prcieability value for most instances is lower bounded by $\frac{k}{k+1}$: essentially, for lower values the entire electorate would be able to afford more than $k$ candidates at rank $m$, hence violating priceability. 

\begin{figure}
    \centering
    \includegraphics[width=.9\linewidth]{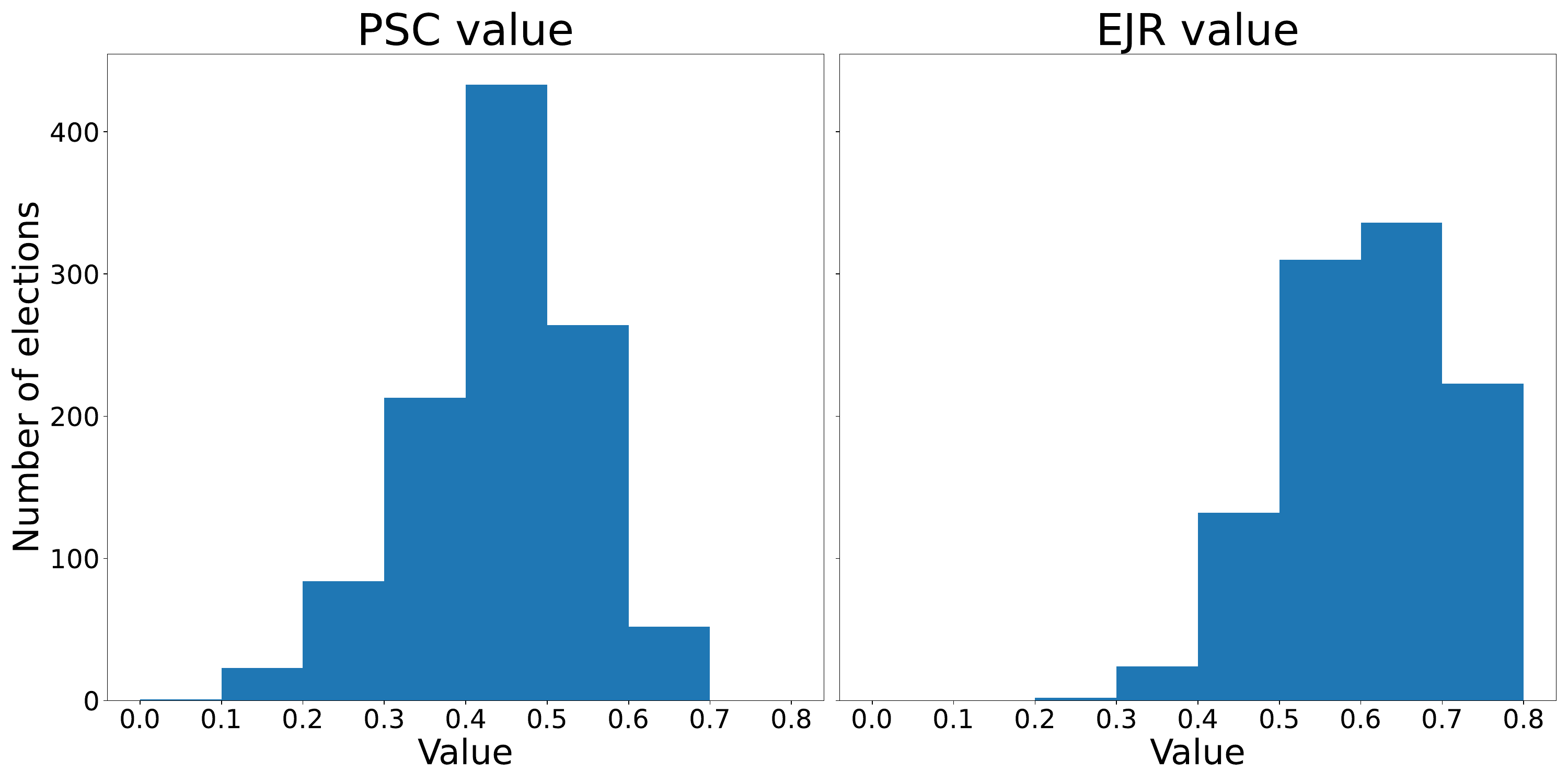}
    \includegraphics[width=.9\linewidth]{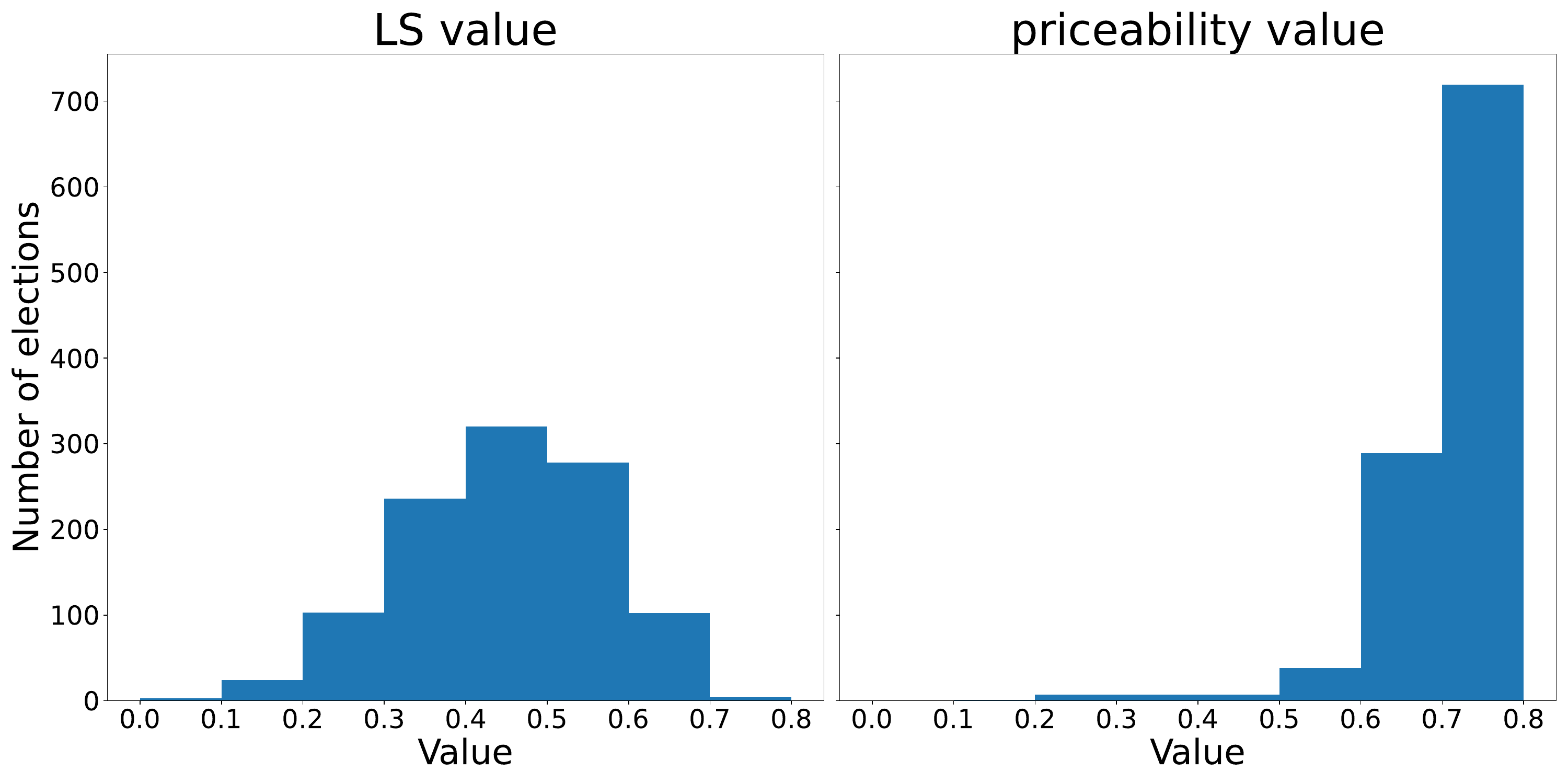}
    \caption{Histograms of optimal values, rounded to one decimal place (completed preferences).}
    \label{fig:hist-complete}
\end{figure}

\subsection{How often do rules produce optimal committees?}
For each voting rule, we counted how often the rule achieves each optimal $\alpha$-value. The results can be found in \Cref{tab:optimal-ptage-avg-dist-complete}, where, for each measure, the value associated with the best-performing rule is highlighted. The overall performance of the rules was slightly worse on the completed data. As in the truncated case, SNTV is most aligned with PSC and LS. On the other hand, on the full preferences, EAR is most aligned with EJR+ and priceability, and now outperforms Scottish STV and Meek STV. Note also that EAR performs better w.r.t. priceability value on the completed data than on the truncated data. 

\begin{table}
    \centering
\begin{tabular}{@{}lcccc@{}}
\toprule
 & PSC & EJR+ & Priceab.& \phantom{X}LS\phantom{X} \\ 
\midrule
S-STV & 763 & 591 & 667 & 743 \\ %
M-STV & 748 & 637 & 675 & 700  \\ %
EAR & 584 & \textbf{{645}} & \textbf{739} & 579 \\ %
SNTV & \textbf{860} & 454 & 543 & \textbf{906} \\ 
seq-RCV & 409 & 618 & 559 & 336 \\ 
\bottomrule
\end{tabular}%
    \caption{For each rule and each axiom, the number of instances for which the rules achieves the optimal $\alpha$-value for completed preferences. The best values in each column appear in bold}
    \label{tab:optimal-ptage-avg-dist-complete}
\end{table}

\subsection{Values Achieved by Rules}
\label{app:comp-boxplots}
We compared the spread of $\alpha$-values achieved by rules to the optimal $\alpha$-values for each measure. The results are presented in \Cref{fig:psc-rules-completed} to \Cref{fig:price-rules-completed}. We observe that the $\alpha$-values of all rules become slightly higher in the complete case – this is in line with the minimal possible values increasing somewhat. Furthermore, the number of seq-RCV committees that have a value of $\alpha > 1$ increases for each measure. In other words, for each axiom, there are more instances where seq-RCV returns a committee that does not satisfy the axiom. We also observe a difference in the spread of priceability values for all rules, compared to that in the truncated case. This has to do with the $\frac{k}{k+1}$ lower bound on the priceability value: the priceability values achieved by rules cluster around this area, to the point where any priceability value that does not is considered an outlier.

\begin{figure}
    \centering
    \includegraphics[width=.95\linewidth]{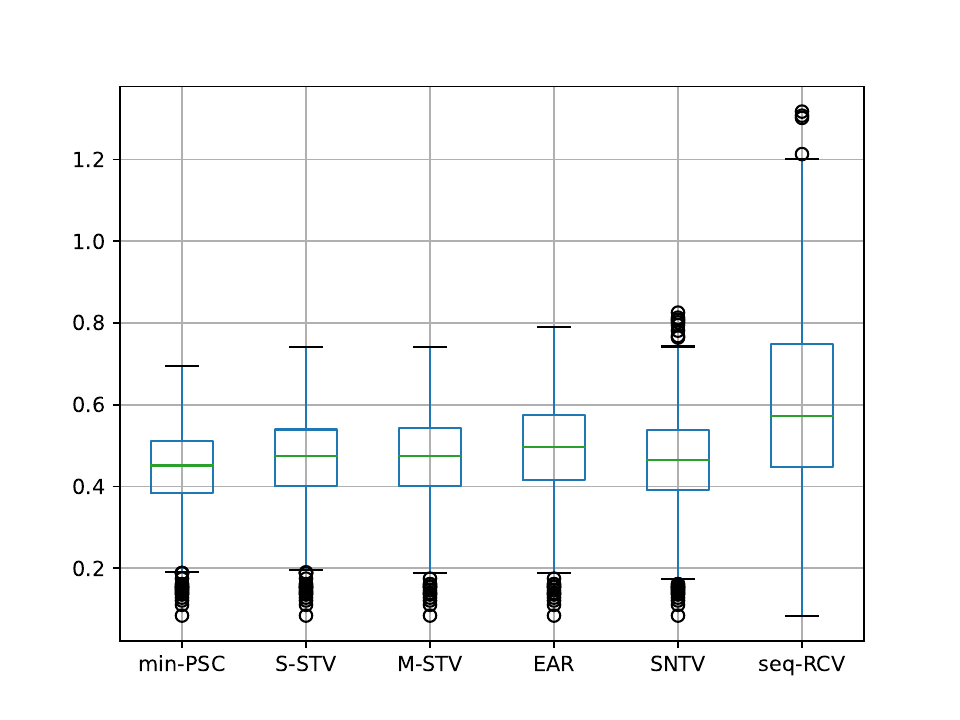}
    \caption{The PSC values achieved by voting rules (completed preferences).}
    \label{fig:psc-rules-completed}
\end{figure}

\begin{figure}
    \centering
    \includegraphics[width=.95\linewidth]{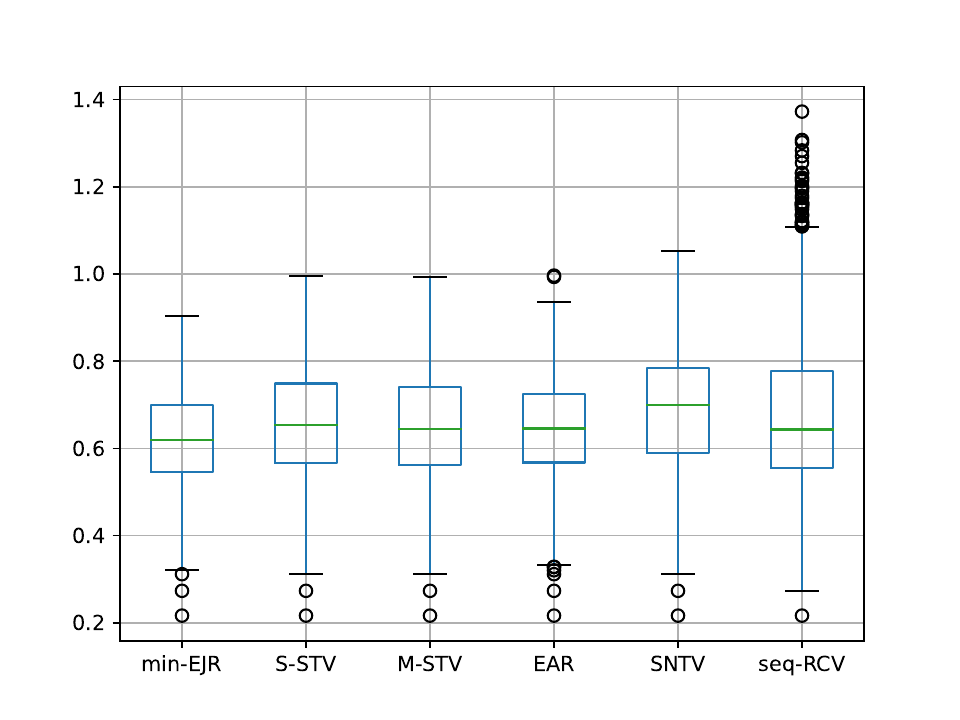}
    \caption{The EJR+ values achieved by voting rules (completed preferences).}
    \label{fig:ejr-rules-completed}
\end{figure}

\begin{figure}
    \centering
    \includegraphics[width=.95\linewidth]{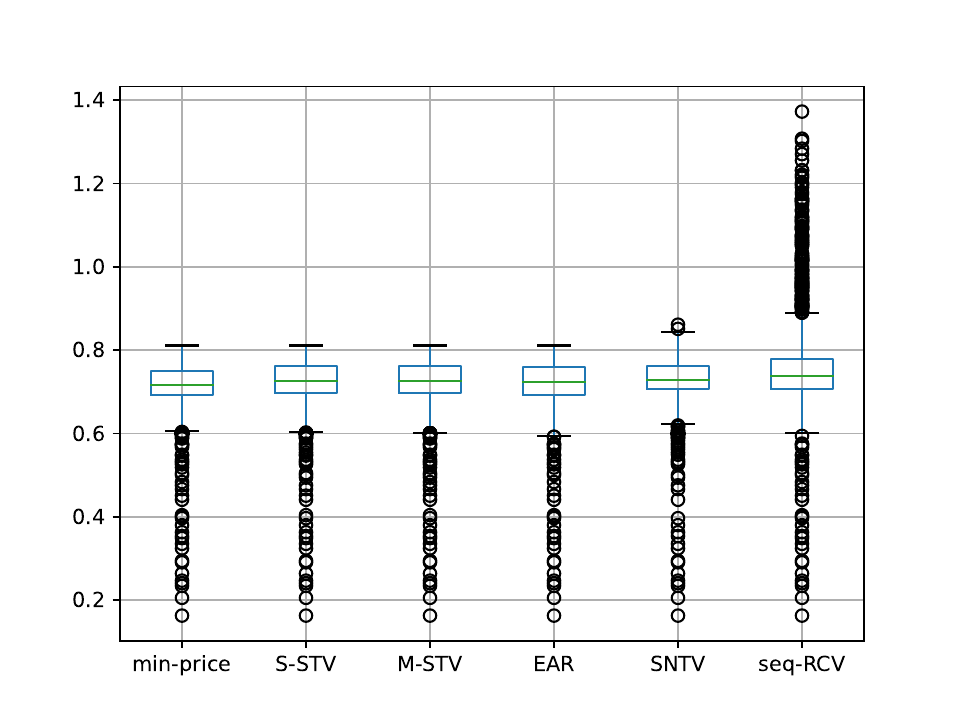}
    \caption{The priceability values achieved by voting rules (completed preferences).}
    \label{fig:price-rules-completed}
\end{figure}

\begin{figure}
    \centering
    \includegraphics[width=.95\linewidth]{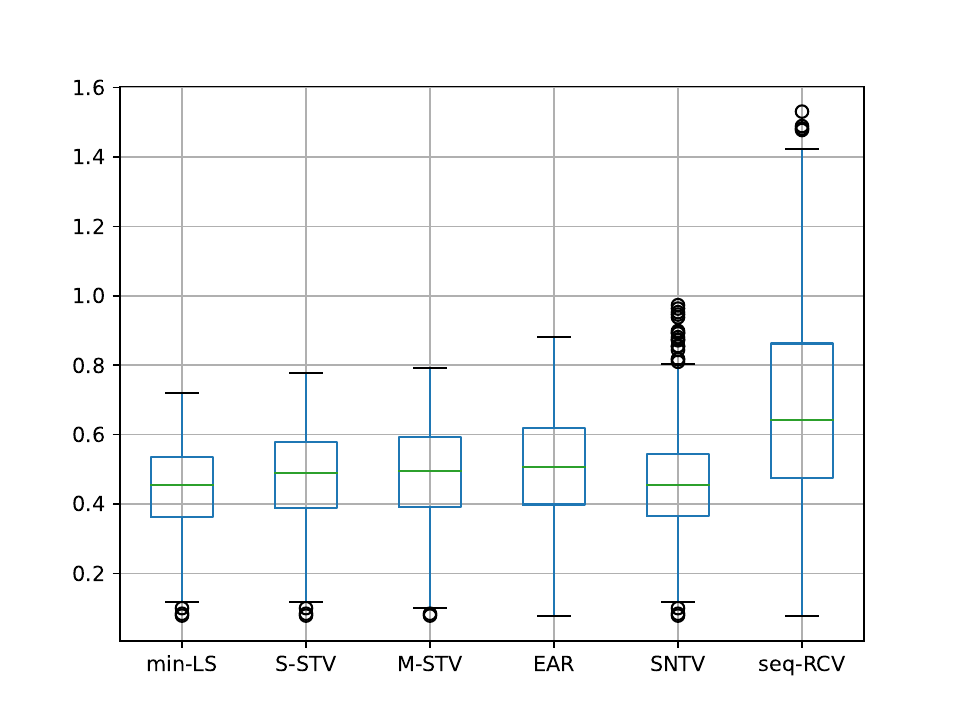}
    \caption{The LS values achieved by voting rules (completed preferences).}
    \label{fig:ls-rules-completed}
\end{figure}

\subsection{Pairwise Comparisons}
\label{app:pw}

We again considered the disagreement and distance between rules. The results are given in \Cref{tab:rules-num-disagreement-complete,tab:rules-avgdist-complete}. On the completed data, S-STV and M-STV agree slightly more often than in the truncated case. Aside from this, disagreement is higher than on the truncated data for all pairs of rules. This difference is perhaps most noticeable for pairwise comparisons that include seq-RCV, which disagrees with every other rule on $t > 100$ more instances than in the experiments on the truncated data.

Furthermore, we had a closer look at PSC values for pairs of rules, restricted to instances on which the rules disagree. We again focused on S-STV vs. seq-RCV (\Cref{fig:stv-rcv-diff-complete}) and \mbox{S-STV} vs. EAR (\Cref{fig:stv-ear-diff-complete}). As was the case on the truncated data, the difference between seq-RCV and S-STV is much more pronounced than what is the case for S-STV and EAR. As mentioned previously, seq-RCV fails 1-PSC more often for completed ballots.

\begin{table}[tb]
    \centering
\begin{tabular}{@{}lccccc@{}} 
\toprule 
 & M-STV & EAR  & SNTV  & seq-RCV \\ %
 \midrule
S-STV & 92 & 306 & 376 & 601 \\ %
M-STV   & -- & 303 & 399 & 555 \\ %
EAR   &   & -- & 536 & 570 \\ %
SNTV    &   &   & -- & 731 \\ %
\bottomrule
\end{tabular}
    \caption{Number of instances pairs of rules disagree on (completed preferences). }
    \label{tab:rules-num-disagreement-complete}
\end{table}

\begin{table}[tb]
    \centering
\begin{tabular}{@{}lcccc@{}} 
\toprule 
 & M-STV & EAR  & SNTV  & seq-RCV \\ %
 \midrule
S-STV &  0.08 & 0.30 & 0.35 & 0.60 \\ %
M-STV     & -- & 0.30 & 0.38 & 0.50 \\ %
EAR    &    & -- &  0.53 & 0.56 \\ %
SNTV    &   &   & -- & 0.76  \\ %
\bottomrule
\end{tabular}
    \caption{Average distance between pairs of rules as fraction of seats (completed preferences).}
    \label{tab:rules-avgdist-complete}
\end{table}

\begin{figure}[t]
    \centering
    \includegraphics[width=\linewidth]{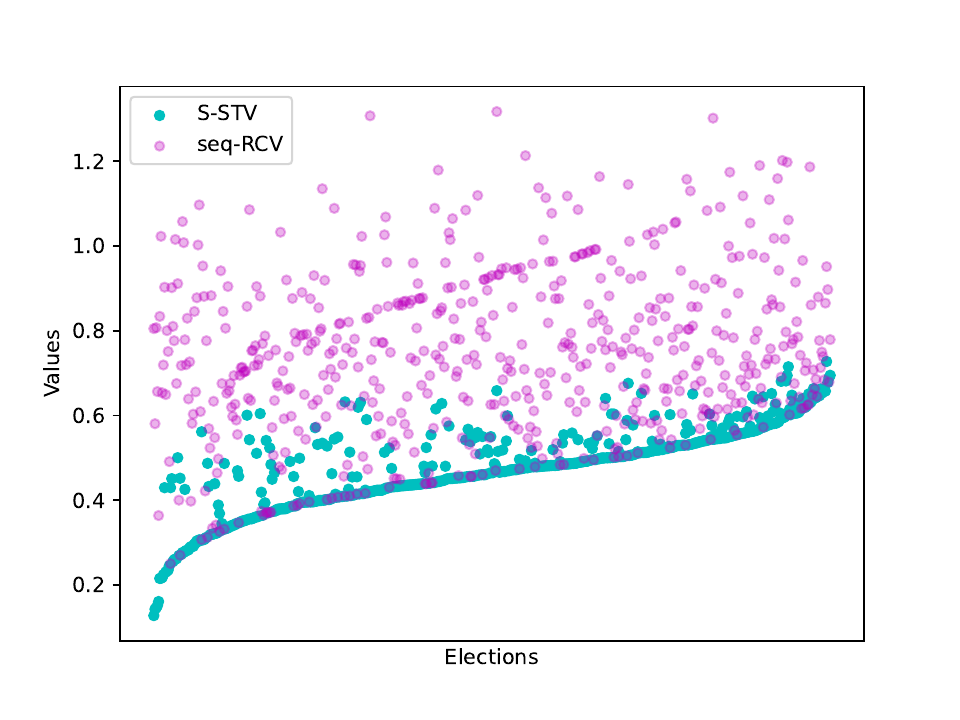}
    \caption{Differences in PSC values between Scottish STV and Sequential RCV when elections where the rules agree are excluded (completed preferences). Elections are ordered by increasing optimal PSC value.}
    \label{fig:stv-rcv-diff-complete}
\end{figure}

\begin{figure}[t]
    \centering
    \includegraphics[width=\linewidth]{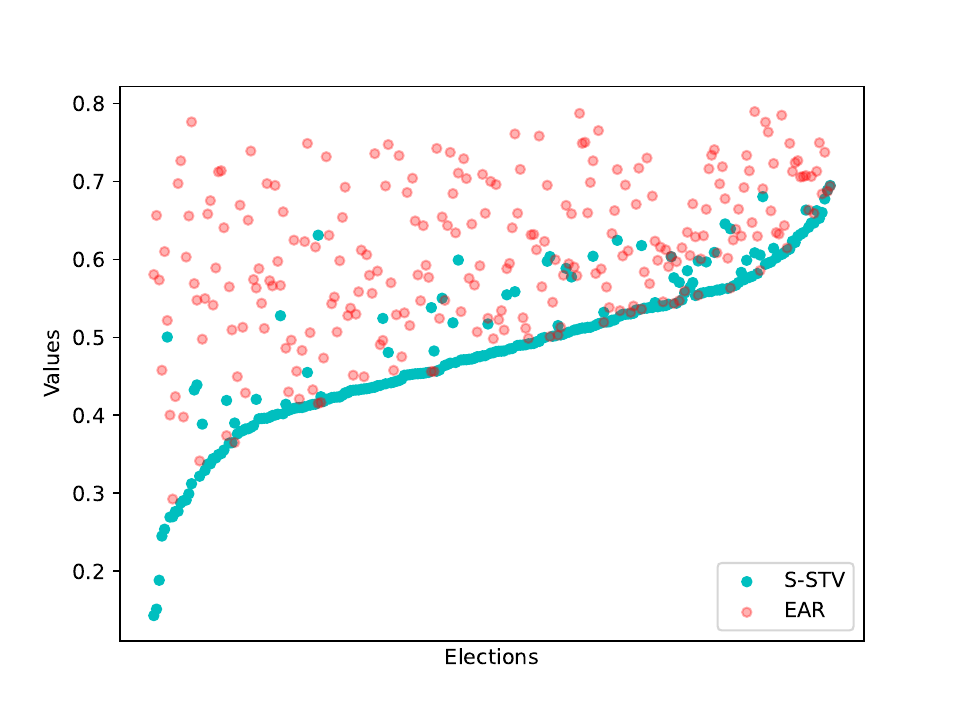}
    \caption{Differences in PSC values between Scottish STV and EAR when elections where the rules agree are excluded (completed preferences). Elections are ordered by increasing optimal PSC value.}
    \label{fig:stv-ear-diff-complete}
\end{figure}
\end{document}